\documentclass[%
 10pt, aps, pra,
 superscriptaddress,
 nofootinbib,
 amsmath,amssymb,
 floatfix,
]{revtex4-2}

\usepackage{graphicx}
\usepackage{amsthm}
\usepackage{xcolor}
\usepackage{bbm}
\usepackage[linesnumbered,ruled,vlined]{algorithm2e}
\usepackage{hyperref}
\usepackage[nameinlink,capitalise,noabbrev]{cleveref}
\usepackage{braket}
\usepackage{mathtools}
\usepackage{bm}
\usepackage[position=top,justification=raggedright,singlelinecheck=false,caption=false]{subfig}

\graphicspath{ {./figures/} }

\newtheorem{theorem}{Theorem}
\newtheorem{lemma}[theorem]{Lemma}
\newtheorem{proposition}[theorem]{Proposition}
\newtheorem{definition}[theorem]{Definition}
\newtheorem{assumption}[theorem]{Assumption}

\newcommand \onecolumnfigurewidth {0.618\textwidth}

\newcommand \threecolumnfigurewidth {0.32\textwidth}

\newcommand{\bb}[1]{\mathbb{#1}}
\newcommand{\ca}[1]{\mathcal{#1}}

\DeclareMathOperator{\poly}{poly}
\DeclareMathOperator{\polylog}{poly\,log}
\DeclareMathOperator{\argmin}{arg\,min}
\DeclareMathOperator{\re}{Re}
\DeclareMathOperator{\im}{Im}

\allowdisplaybreaks

\hypersetup{colorlinks, 
citecolor=blue,
linkcolor=red,
urlcolor=purple}

\begin{document}

\title{Classical combinations of quantum states for solving banded circulant linear systems}

\author{Po-Wei Huang}
\email{huangpowei22@u.nus.edu}
\affiliation{Centre for Quantum Technologies, National University of Singapore, Singapore 117543}

\author{Xiufan Li}
\email{lixiufan@u.nus.edu}
\affiliation{Centre for Quantum Technologies, National University of Singapore, Singapore 117543}

\author{Kelvin Koor}
\affiliation{Centre for Quantum Technologies, National University of Singapore, Singapore 117543}

\author{Patrick Rebentrost}
\email{cqtfpr@nus.edu.sg}
\affiliation{Centre for Quantum Technologies, National University of Singapore, Singapore 117543}
\affiliation{Department of Computer Science, National University of Singapore, Singapore 117417}

\date{\today}

\begin{abstract}
Solving linear systems is of great importance in numerous fields. Proposed quantum algorithms for preparing solutions for linear systems include the HHL algorithm with subsequent refinements and variational methods. 
Circulant linear systems appear in many physics-related differential equations. An interesting case is banded circulant linear systems whose non-zero terms are within distance $K$ of the main diagonal.
For these systems, we propose an approach based on the classical combination of quantum states (CQS) method relying on convex optimization against the available analytical solution. From decompositions into cyclic permutations, the solution can be approximately represented by a classical combination of a polynomial number of quantum states. We validate our methods using classical simulations as well as execution on an IBM quantum computer. While in the setting of this paper, efficient classical algorithms are available, our results demonstrate the potential applicability of the CQS method for solving physics problems such as heat transfer.
\end{abstract}

\maketitle

\section{Introduction}
Quantum algorithms have been designed to find approximate solutions to linear systems of equations. \citet{harrow2009quantum} gave an efficient quantum algorithm that prepares solutions to sparse and well-conditioned linear systems. However, the HHL algorithm consumes a large amount of quantum resources and will not be useful before the realization of fault-tolerant quantum computers. On the other hand, variational methods~\citep{cerezo2021variational, xu2021variational} such as the variational quantum linear solver (VQLS)~\citep{bravo_prieto2020variational} are shown to be executable on noisy intermediate-scale quantum (NISQ) devices~\citep{preskill2018quantum} but suffer from the barren plateau problem~\citep{mcclean2018barren} as well as a lack of performance guarantees. To address these problems, \citet{huang2021term} proposed an efficient hybrid quantum-classical linear systems solver based on \emph{classical combinations of quantum states} (CQS), Ansatz trees, and the Hadamard test. Given a linear system $Ax=b$ where an efficient quantum state representation $\ket{b}$ can be prepared, as well as a decomposition of $A = \sum_i \beta_i U_i$ into a weighted sum of efficiently implementable operators $A_i$, based on Krylov methods for constructing inverse operators $A^{-1}$, one can construct the solution vector $x$ as a weighted sum of $x = \sum_j \alpha_j W_j \ket{b}$, where $W_j$ is a product of multiple combinations of $U_i$ heuristically found by traversing an Ansatz tree, presenting the solution vector as a \emph{classical combination of quantum states}. By combining measurement results from Hadamard tests~\citep{cleve1998quantum} and classical convex optimization, one can obtain the exact values of $\alpha_j$ required to reconstruct the solution vector $x$. Besides providing theoretical guarantees of finding the approximate solution, the CQS algorithm only requires reasonable amounts of quantum resources (such as qubit count and single-qubit controlled gates). The algorithm saves ancilla qubits and circuit depth by removing the preparation of linear combinations of unitaries compared to the fully coherent version of the algorithm. Hence, this algorithm can be considered an example for the \textit{early fault-tolerant} quantum era~\citep{campbell2021early, dong2022groundstate, wan2022randomized, wang2024qubitefficient, katabarwa2024early}.

Solving \emph{circulant} linear systems poses a fundamental problem in numerous fields such as modelling physical systems, signal processing, and image reconstruction. The special properties of these systems permit fast diagonalization via the fast Fourier transform (FFT)~\citep{davis1979circulant}, requiring only $\mathcal O (N\log N)$ operations~\citep{cooley1965algorithm} where $N$ is the system size), compared to $\mathcal O (N^3)$ via Gaussian elimination. Furthermore, the conjugate gradient method takes $\mathcal O (N\kappa)$ operations~\citep{hestenes1952methods}, where $\kappa$ is the condition number. Even faster classical algorithms based on direct applications of the Krylov method have been designed~\citep{saad1981krylov}. In terms of quantum algorithms, \citet{zhou2017efficient} have demonstrated an efficient implementation of circulant matrices through unitary decomposition and the quantum Fourier transform (QFT)~\citep{coppersmith1994approximate}, allowing the problem to be solved using quantum linear system solvers. \citet{wan2018asymptotic} directly constructed the inverse of the circulant matrix by diagonalizing via QFT and amplitude amplification~\citep{brassard2002quantum}. 

In this work, we adapt the classical combination of quantum states (CQS) method by \citet{huang2021term} in the context of solving circulant linear systems. Circulant systems provide substantial technical simplifications in the CQS algorithm and reduce the number of circuits required for evaluation. Precisely because of this simplicity, a classical algorithm that analytically computes the coefficients obtained in the CQS algorithm can also be constructed. Consequently, we compare the results of the CQS method against the analytical results obtained from the classical algorithm. Although the problem we investigate in this paper is classically tractable, the paper aims to provide a proposal of methodologies regarding quantum algorithms in the early fault-tolerant regime that warrants further investigation. 

Following~\citet{huang2021term}, instead of explicitly diagonalizing the matrix with the discrete Fourier transform (DFT), we first decompose the circulant matrix into a sum of cyclic permutation operators. These operators can then be diagonalized individually by DFT. Using QFT in our algorithm requires quantum hardware beyond the capacities of current near-term systems~\citep{nam2020approximate}. Compared to the naive application of CQS, our theoretical results show that such a decomposition provides a much tighter upper bound for convergence guarantees, allowing an experimentally achievable guarantee without invoking gradient expansion Ansatz trees and other heuristical methods, or additional regularization such as Tikhonov regularized regression on the solution vector.

The numerical performances of our algorithms are demonstrated through numerical solutions\footnote{Code implementation can be found at \url{https://github.com/LiXiufan/qa-cqs-circulant}.} to one-dimensional time-dependent heat transfer problems discretized by finite difference methods~\citep{ozicsik2017finite,thomas1995numerical}. Through classical simulations, we show that both the hybrid algorithm and the quantum-inspired algorithm provide good approximations to the exact solutions obtained by the Krylov method. Furthermore, we implement the hybrid algorithm on IBM quantum hardware with a mean-square-error (MSE)  of less than 0.05, demonstrating the feasibility of solving physical problems on practical quantum computers.

\section{Preliminaries}\label{sectPrelim}
\subsection{Notations.} Let $[n]:=\{0, 1, \dots, n-1\}$ and $[a \dots b]$ be the range of integers from $a$ to $b$, inclusive of endpoints. Given a field $\mathbb{F}$ of real or complex numbers, for vectors $u, v \in \mathbb{F}^N$, we denote their inner product by $\langle u, v \rangle = u^\dagger v$, where $u^\dagger$ is the adjoint of $u$. For quantum states, we use Dirac notation $\braket{u|v}$. We denote a vector's $\ell_2$ norm by $\|v\|:= \sqrt{\sum_{i=1}^N |v_i|^2} = \langle v, v\rangle$. Let $\mathcal M_{N}(\mathbb{F})$ be the space of square matrices of size $N$ on the field $\mathbb{F}$.  For a matrix $A\in \mathcal M_{N}(\mathbb{C})$, let $A_{ij}$ be the $(i, j)$-element of $A$. We denote the spectral norm by $\|A\|:= \sup_{x\ne 0}\frac{\|Ax\|}{\|x\|} = \max_{i}\sigma_i(A)$, and the Frobenius norm by $\|A\|_F = \sqrt{\sum_i\sum_j |A_{ij}|^2} = \sqrt{\sum_{i=1}^{N} \sigma_i^2(A)}$. where $\sigma_i(A)$ are the singular values of $A$. We denote the transpose of $A$ by $A^T$, the adjoint by $A^\dagger$, and the pseudoinverse by $A^{-1}$. Note that $\|A^{-1}\| = \frac{1}{\sigma_{\min}(A)}$, where $\sigma_{\min}(A)$ is the smallest non-zero singular value of $A$. The condition number of $A$ is denoted by $\kappa_A := \|A\|\|A^{-1}\| = \frac{\sigma_{\max}(A)}{\sigma_{\min}(A)}$, and the eigenvalues by $\lambda(A)$. Lastly, we note that quantum circuits in this paper are ordered in the big-endian format, where the least significant qubit is placed at the bottom of the circuit. 

\subsection{Circulant matrices.}
In this paper, we are interested in a special yet useful case of solving $Cx=b$, namely when $C$ is a banded circulant matrix. We begin by defining and discussing the properties of circulant matrices. A more in-depth review of these matrices can be found in~\citep{davis1979circulant,gray2005toeplitz, chen1987solution}.
\begin{definition}[Circulant matrix]
Let $\mathbb{F}$ = $\mathbb{R}$ or $\mathbb{C}$. A matrix $C \in \mathcal M_{N}(\mathbb{F})$ is called a \textit{circulant matrix} if it takes the form
\begin{equation*}
    C = \begin{pmatrix}
    c_0 & c_{N-1} & c_{N-2} & \cdots & c_1\\ 
    c_1 & c_0 & c_{N-1} & \cdots & c_2\\ 
    c_2 & c_1 & c_0 & \cdots & c_3\\ 
    \vdots & \vdots & \vdots & \ddots & \vdots\\
    c_{N-1} & c_{N-2} & c_{N-3} & \cdots & c_0    \end{pmatrix}.
\end{equation*}
\end{definition}
In this matrix, each column (row) vector is the previous column (row) cyclically shifted by one step. Any circulant matrix $C$ can be written as a linear combination of \textit{cyclic permutation matrices} $Q^\ell \in \mathcal M_{N}(\mathbb{R})$:
\begin{equation}
\label{eqPermute}
    C = \sum_{\ell=0}^{N-1} c_\ell Q^\ell, \quad  Q = \begin{pmatrix}
    0 & 0 & 0 & \cdots & 0 & 1\\ 
    1 & 0 & 0 & \cdots & 0 & 0\\ 
    0 & 1 & 0 & \cdots & 0 & 0\\
    \vdots & \vdots & \vdots & \ddots & \vdots & \vdots\\
    0 & 0 & 0 & \cdots & 1 & 0\\ \end{pmatrix}.
\end{equation}
One can see that $C$ is normal, unitary, and diagonalizable. The $k$-th eigenvalue of $Q$ is given by $\lambda_k(Q) = \omega_N^k$, for all $k \in [N]$, where we have adopted the notation $\omega_N^k := e^{2\pi ik/N}$. The $k$-th eigenvalue of circulant matrix $C$ is thus $\lambda_k(C) = \sum_{\ell=0}^{N-1}c_\ell\omega_N^{k\ell}$ for all $k \in [N]$. Observe that the $k$-th eigenvector of $C$ is $x^{(k)} = \frac{1}{\sqrt{N}} \begin{pmatrix}1&\omega_N^{k}&\omega_N^{2k}&\cdots&\omega_N^{(N-1)k}\end{pmatrix}^T$. Putting them together to form the diagonalizing matrix $F$, we see that $F_{jk} = \frac{1}{\sqrt{N}}\omega_N^{jk} = \frac{1}{\sqrt{N}}e^{2\pi ijk/N}$, or more explicitly,
\begin{equation}
    F = \frac{1}{\sqrt{N}}
    \begin{pmatrix}
         & \vdots & \\
        \dots & \omega_N^{jk} & \dots\\
         & \vdots & 
    \end{pmatrix}.
\end{equation}
This is the matrix implementing the discrete Fourier transform (DFT), which is used to diagonalize circulant matrices.

\subsection{Banded circulant matrices.}
We now define the banded circulant matrix we aim to address in our paper.
\begin{definition}[$K$-banded circulant matrix~\citep{chen1987solution}]\label{defBandedCirculantMatrix}
Given $K \in \mathcal O(\polylog N)$, a $K$-banded circulant matrix is a ``sparse'' circulant matrix where $c_{K+1}, c_{K+2}, \cdots, c_{N-K-1} = 0$:
\begin{equation*}
C=\setlength\arraycolsep{1pt}\begin{pmatrix} 
c_0 & c_{N-1} &  \cdots & c_{N-K}   & &  &   c_K & \cdots & c_1\cr
c_1 & c_0 &\ddots &  & \ddots     &  & & \ddots & \vdots\cr
\vdots & \ddots  & \ddots &\ddots   & &  \ddots     &  & & c_K \\
c_K & &\ddots  &\ddots  & \ddots &  & \ddots  & & \\
 &    \ddots  &   &\ddots   &\ddots   & \ddots  &   & \ddots  &  \\
 &  &   \ddots  &   &\ddots   &\ddots   & \ddots  &   & c_{N-K}  \\
c_{N-K} &   &   &   \ddots   &   &\ddots  &\ddots   & \ddots  & \vdots \\
\vdots & \ddots   &   &   &   \ddots   &   &\ddots  & c_0  &c_{N-1}  \\
c_{N-1} & \cdots  & c_{N-K}   &   &   &    c_K  & \cdots  & c_1  & c_0
\end{pmatrix}
\end{equation*}
\end{definition}
Noting that $Q^{N-\ell} = Q^{-\ell}$, we define $c_{-\ell} := c_{N-\ell}$ for $\ell\in[N]$. Therefore, a $K$-banded circulant matrix can be expressed as a degree-$2K+1$ matrix polynomial:
\begin{equation}
    C = \sum_{\ell=-K}^{K} c_\ell Q^\ell.
\end{equation}

\section{Solving banded circulant linear systems}
\subsection{Problem statement and solution.}
Let us be given a $K$-banded circulant matrix $C\in \mathcal M_{N}(\mathbb{C})$ and a normalized vector $\bm b\in \mathbb{C}^N$. Our task is to find a good estimator $\tilde{x}$ for the least squares solution $x^*=C^{-1}\bm b$, such that the mean-square-error (MSE) loss of $\tilde{x}$, or $\|C\tilde{x}-\bm b\|^2$, is close to the MSE loss of $x^*$, or $\|Cx^*-\bm b\|^2$. More precisely, given additive error tolerance $\varepsilon$, we aim to find $\tilde{x}$ such that
\begin{equation}\label{eqMainProblem}
    \mathcal{L}_{{\rm MSE}}(\tilde x) := \|C\tilde x -\bm b\|^2 \le \|Cx^* -\bm b\|^2 + \varepsilon = \min_{x\in\mathbb C^N}\|C x -\bm b\|^2 + \varepsilon = \min_{x\in\mathbb C^N}\mathcal L_{{\rm MSE}} (x) + \varepsilon.
\end{equation}

By using Chebyshev polynomials to the $\mathcal{O}(K \cdot \kappa_C\log\frac{\kappa_C}{\nu})$-th degree~\citep{childs2017quantum}, it is possible to obtain a closed-form solution for approximating $x^* = C^{-1}\bm b$ such that $\mathcal L_{{\rm MSE}}(\tilde x)$ is $\nu$-close to $\mathcal \min_{x\in\mathbb C^N} \mathcal L_{{\rm MSE}}(x)$ as stated by the proposition below. The proof can be found in \cref{appendixGuarantee}. 

\begin{proposition}\label{propGuarantee}
Let $0 < \nu \le 1$. Given a $K$-banded circulant matrix $C\in \mathcal M_{N}(\mathbb{C})$ where $\kappa_C = \|C\|\|C^{-1}\|$ and a normalized vector $\bm b\in \mathbb{C}^N$, there exists $T\in \mathcal{O}(K \cdot \kappa_C\log\frac{\kappa_C}{\nu})$ such that we can find an optimal set of parameters $\alpha = \{\alpha_m \in \mathbb{C}, \forall m\in[-T.. T] \}$ and that the estimator
$\tilde{x}(\alpha) = \sum_{m=-T}^T \alpha_{m}Q^m\bm b$ satisfies
\begin{equation*}
    \min_{\alpha \in \mathbb{C}^{2 T + 1}} \|C\tilde{x}(\alpha) - \bm b\|^2 \le \min_{x \in \mathbb{C}^N} \|Cx - \bm b\|^2 + \nu.
\end{equation*}
\end{proposition}

The constructive proof of the algorithm gives way to a classical implementation that would provide a Chebyshev expansion-esque construction of the solution of $x$ via direct computation of the Chebyshev coefficients. In this paper, we focus on comparing the results of the CQS method against the analytical results that can be obtained from such a classical algorithm or simply conducting matrix inversion, instead of providing a particular use case where the hardness of the problem is raised to a significant level, where quantum advantage can be proposed.

We note that the above proposition provides a tighter bound on the number of quantum states required compared to the $(2K+1)^{\mathcal O(\kappa_C\log\frac{\kappa_C}{\nu})}$ upper bound of a full Ansatz tree obtained by~\citet{huang2021term}. 

\begin{figure}
\includegraphics[width=\textwidth]{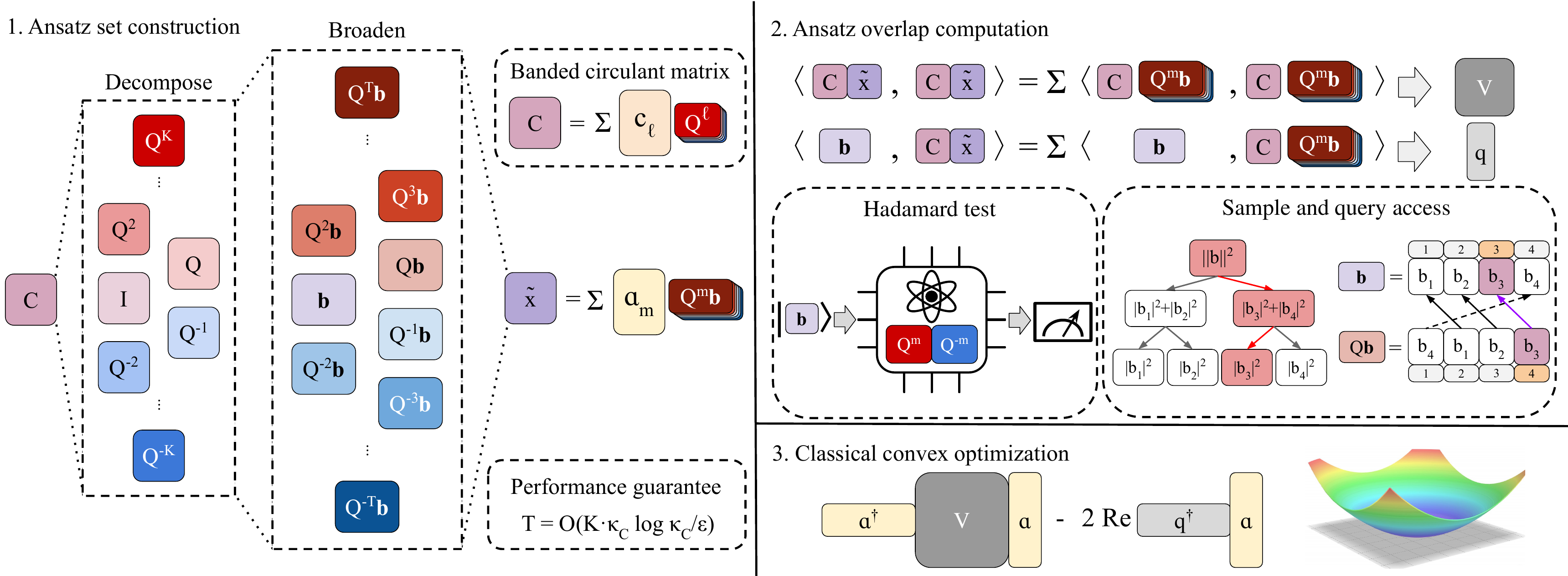}
\caption{Here we provide an illustrated overview of our method to solve $K$-banded circulant linear systems. Note the fact that $C$ can be written as a matrix polynomial in terms of the cyclic permutation matrix $Q$ given in \cref{eqPermute} with the powers being bounded between $-K$ and $K$. In step (1), we set a truncation threshold $T\ge K$ to create a matrix polynomial with powers of $Q$ from $-T$ to $T$ that serves as an approximation to $C^{-1}$. To find the optimal coefficients of the matrix polynomial $\alpha$, in step (2), we calculate the overlaps between individual components $CQ^m\bm b$ as well as the overlap between $CQ^m\bm b$ and $\bm b$. Such overlaps can be computed with either the Hadamard test for the hybrid quantum-classical algorithm or sample and query access for the quantum-inspired algorithm. Lastly, in step (3), we perform classical convex optimization to find the optimal coefficients $\alpha$.}
\label{figMain}
\end{figure}

\subsection{CQS algorithm.} 
We employ the CQS technique by~\citet{huang2021term} and use \cref{algoMain} to solve the given circulant linear system of equations. An overview can be found in \cref{figMain}. We outline the idea behind the algorithm below.

\begin{figure}
\begin{algorithm}[H]
\caption{High-level algorithm for solving banded circulant linear systems}
\label{algoMain}
\DontPrintSemicolon
\KwIn{Decomposition coefficients of $K$-banded circulant matrix $\{c_\ell\}_{\ell\in[-K..K]}$, Normalized column vector $\bm b$, Truncation threshold $T$.}
\KwOut{Combination parameters $\{\alpha_m\}_{m\in[-T.. T]}$, Sketches of $\{\bm u_m: \bm u_m = Q^m \bm b, \forall m\in [-T.. T]\}$}
Estimate entry values of matrix $V^R \in \mathcal M_{2T+1}(\mathbb R)$, where $V^R_{jk} = \re\left(\sum\limits_{y=-K}^K\sum\limits_{z=-K}^K c_y^*c_z \braket{\bm b, Q^{y-z+k-j}\bm b}\right)$.\;
Estimate entry values of matrix $V^I \in \mathcal M_{2T+1}(\mathbb R)$, where $V^I_{jk} = \im\left(\sum\limits_{y=-K}^K\sum\limits_{z=-K}^K c_y^*c_z \braket{\bm b, Q^{y-z+k-j}\bm b}\right)$.\;
Estimate entry values of vector $q^R \in \mathbb R^{2T+1}$, where $q^R_{j} = \re\left(\sum\limits_{y=-K}^K c_y \braket{\bm b, Q^{y+j}\bm b}\right)$.\;
Estimate entry values of vector $q^I \in \mathbb R^{2T+1}$, where $q^I_{j} = \im\left(\sum\limits_{y=-K}^K c_y \braket{\bm b, Q^{y+j}\bm b}\right)$.\;
Let $W = \begin{pmatrix} V^R & -V^I \\ V^I & V^R \end{pmatrix}$ and $r = \begin{pmatrix}q^R \\ q^I\end{pmatrix}$.\;
Solve for $z$ in the convex optimization problem $z^T W z - 2r^T z +1$.
$\forall m \in [-T.. T]$, let $\alpha_m = z_m + i \cdot z_{2T+1+m}$.\;
\Return $\{(\alpha_m, Q^m\bm b)\}_{m\in [-T.. T]}$\;
\end{algorithm}
\end{figure}

Per high-level intuition, we see that the estimator $\tilde{x}$ is assembled from a set of quantum states such that given a truncation threshold $T \ge K$, we have the tunable parameter set $\alpha:= \{\alpha_i: \alpha_i \in \mathbb C\}$ and
\begin{equation}
\tilde x(\alpha) = \sum_{m=-T}^{T} \alpha_m \bm u_m = \sum_{m=-T}^{T} \alpha_m Q^m  \bm{b} .
\end{equation}
We minimize the combination parameters $\alpha$ over the MSE loss 
\begin{equation}
 \mathcal{L}_{{\rm MSE}}(\tilde x) =  \|C\tilde{x} - \bm b\|^2 = \braket{C \tilde x, C \tilde x} - 2 \re \{\braket{\bm b, C\tilde{x}}\} + 1,
\end{equation}
to achieve an approximation to $x^*$. Following \citep{huang2021term}, we construct the complex matrix $V\in \mathcal M_{2T+1}(\mathbb C)$ and complex vector $q \in \mathbb{C}^{2T+1}$ with
$V_{jk} := \braket{C\bm u_j, C\bm u_k}$ and $q_j := \braket{\bm b,C\bm u_j}$ such that
\begin{equation}\label{eqSolve}
\|C\tilde{x} - \bm b\|^2 = \alpha^\dagger V \alpha - 2 \re \{ q^\dagger \alpha \} + 1.
\end{equation}
We note that the terms $V_{jk}$ and $q_j$ can be computed from the summation of inner products of the form $\braket{b, Q^jb}$ for $k\in [-2T-2K.. 2T+2K]$:
\begin{align}\label{eqVjk}
V_{jk} = \braket{C \bm u_j, C\bm u_k} &= \sum_{y=-K}^K\sum_{z=-K}^K c_z^*c_y \braket{\bm b,Q^{y-z+k-j}\bm b},\\
q_j = \braket{\bm b, C\bm u_j} &= \sum_{y=-K}^K c_y \braket{\bm b, Q^{y+j}\bm b}.
\end{align}

To solve for the combination parameters $\alpha$, we define the auxiliary system matrix $W  \in \mathcal{M}_{4T+2}(\mathbb{R})$ and vector $r \in \mathbb{R}^{4T+2} $: 
\begin{equation}\label{eqAuxiliary}
W := \begin{pmatrix} \re\{V\} & - \im\{V\} \\ \im\{V\} & \re\{V\} \end{pmatrix}, \quad r := \begin{pmatrix}\re\{q\}\\\im\{q\}\end{pmatrix}.
\end{equation}
We can then recast the above problem into a convex optimization problem involving a real vector $z \in \mathbb{R}^{4T+2}$ with $z:= (\re \{\alpha\}, \im \{\alpha\})^{T}$. The problem becomes
\begin{equation}
    \min_{z\in \mathbb{R}^{4T+2}} z^T Wz - 2 r^T z + 1.
\end{equation}
Solving for $z$ gives us the real and imaginary parts of all coefficients in the set $\{\alpha_m\}_{m=-T}^T$. Recalling that our goal is to obtain an estimator $\tilde x$ of $x^*$, we can represent $\tilde x$ by combining the coefficients $\alpha$ and classical sketches of the vectors $\{\bm u_m = Q^m \bm b\}_{m=-T}^T$. Note that we do not obtain the explicit value of $\tilde x$. We will address such limitations at the end of this section.

\subsection{Limitations.}
As mentioned, \cref{algoMain} provides $\tilde x$ as a weighted combination of vectors that exist in the direct sum of the Krylov subspaces $\mathcal{K}_T (Q, \bm b)$ and $\mathcal{K}_T (Q^{-1}, \bm b)$ instead of an explicit vector. Such representations are useful for investigating certain properties of $x$, such as the expectation $x^\dagger M x$ for some operator $M$, but the actual vector cannot be retrieved without an exponential-cost post-processing procedure. Such limitations have also been found~\citep{aaronson2015read} in the HHL algorithm~\citep{harrow2009quantum}, which produces the solution as a quantum state $\ket x$, with the retrieval of the explicit vector $x$ requiring exponential cost as well.

For the output of \cref{algoMain}, note that we output both the coefficients $\{\alpha_m\}_{m=-T}^T$ and the corresponding Ansatz set $\{\bm u_m\}$ to form a representation of $\tilde x$. For the quantum algorithm, quantum circuits that prepare $\ket{u_m}$ are prepared, while for the quantum-inspired algorithm, data structures that provide sample and query access to vectors $\bm u_m$ are provided. We detail the preparation of these accesses in \cref{sectHybrid}.

Finally, \cref{propGuarantee} guarantees that the loss of the solution that \cref{algoMain} provides will converge if the truncation threshold $T\in\mathcal O(K\cdot\kappa_c\log\kappa_c/\varepsilon)$. Hence, for ill-conditioned banded circulant matrices where $\kappa_C \to \infty$, we do not provide an effective guarantee of convergence, which also occurs in classical Krylov subspace-based algorithms such as the conjugate gradient method~\citep{hestenes1952methods}. However, in practical scenarios, the algorithm may still converge; \cref{propGuarantee} only provides an upper bound for convergence guarantees. Convergence may still occur with truncation thresholds lower than the upper bound limit, a trait observed in our numerical results in \cref{sectExp}.

\section{Hybrid algorithm and quantum circuit implementation}~\label{sectHybrid}
We now elaborate on the implementation of \cref{algoMain} using quantum hardware. Without loss of generality, let $n = \lceil\log_2 N\rceil$. Given an $n$-qubit quantum system, to implement the algorithm, we require the following assumption:
\begin{assumption}\label{assumptionHardware}
Assume that there exists a hardware-efficient $n$-qubit quantum circuit evaluating unitary $U_b$ such that the normalized vector $\bm b$ can be implemented as a quantum state $\ket{b} = U_b\ket{0^n}$ with $\mathcal O (\poly n)$ circuit depth.
\end{assumption}

We note that this assumption of efficient preparations of  $\ket{b}$ may seem strong, but this is the nature of various linear system solvers, including CQS itself~\citep{huang2021term}, and fault-tolerant linear solvers such as HHL~\citep{harrow2009quantum}. In many cases of solving partial differential equations in a periodic boundary setting, however, the state $\ket{b}$ may admit a simple representation in discrete space, for example, a point charge in phase space. Further, preparation of Gaussian states, a common initial state for PDEs and circulant systems, has been shown to be polynomially scaling to the number of qubits~\citep{grover2002creating,kitaev2009wavefunction}. In other cases where there may be an analytic expression for the initial state, one can combine quantum arithmetic circuits with controlled rotations to prepare the state without full amplitude embedding, and reduce the cost of state preparation to $\mathcal O (\poly n)$. As the construction of our algorithm already requires the usage of QFT, the usage of quantum arithmetic circuits based on manipulations in Fourier space is no harder than the circuits we use to construct the $Q$ operator, which can also be viewed as a unitary adder. 

We implement steps 1-4 of \cref{algoMain} with quantum computers, obtaining estimations of each entry of $V^R$, $V^I$, $q^R$, and $q^I$ with the Hadamard test~\citep{cleve1998quantum}, which are then stored in classical memory. We modify quantum circuit implementations of shift operations shown by \citet{koch2021gate} and utilize QFT-based adder circuits~\citep{beauregard2003circuit, draper2000addition} to implement the cyclic permutation operator $Q$ and its powers with a depth independent of the power $m$, a contrast to the implementation of quantum circuits for solving general linear systems by~\citet{huang2021term}. Steps 5-8 of the algorithm remain implemented on classical computers. As shown in the following section, the powers of $Q$ can be implemented with a known quantum circuit with $\mathcal O(n^2)$ quantum gates. Therefore, we can construct classical sketches of the circuit implementation of output states $\ket{u_m} = Q^mU_b\ket{0^n}$ that correspond to the computed optimal parameters $\alpha_m$, using $\mathcal O(Tn^2)$ memory in total to provide a classical sketch of the estimator $\tilde x$ as a combination of quantum states. We now provide details of the circuit implementation as follows:

\subsection{Quantum circuit implementation.} Recall from \cref{sectPrelim} that the eigendecomposition of the cyclic permutation matrix $Q$ is given by $Q=F^{-1}\Lambda F$, where $\Lambda$ is the diagonal matrix with all eigenvalues of $Q$ and the diagonalization matrix $F$ is the matrix representation of DFT. On quantum systems, for cases where $n= \log_2(N)$, $F$ can be implemented with $\mathcal{O}(n^2)$ quantum logic gates using the quantum Fourier transform (QFT)~\citep{coppersmith1994approximate}. For the more general case where $n \ne \log_2 N$, one can use the arbitrary-size quantum Fourier transform proposed by \citet{kitaev1995quantum} which utilizes quantum phase estimation to construct DFT. For the remainder of the discussion, we assume $n= \log_2(N)$, while keeping in mind that cases where $n\ne \log_2(N)$ are still implementable. Matrix $\Lambda$ can be written as:
\begin{equation}
\Lambda = \operatorname{diag}\begin{pmatrix}
    \omega_N^0 & \omega_N^1 & \omega_N^2 & \cdots & \omega_N^{N-1}
\end{pmatrix}
\end{equation}
We use the phase gate $P$, whose matrix form is
\begin{equation}
    P(\theta) = \begin{pmatrix}
    1 & 0\\
    0& e^{i \theta}\\
\end{pmatrix}, 
\end{equation}
where the parameter $\theta$ is the rotation angle. We note that $\omega^0 = 1$. Hence $\Lambda$ can be written as the tensor product of $P(\theta)$ gates as follows:
\begin{equation}
\Lambda = 
\begin{pmatrix}
    1 & 0\\
    0&\omega^{N/2}_N\\
\end{pmatrix}
\otimes\cdots\otimes
\begin{pmatrix}
    1 & 0\\
    0&\omega^{2}_N\\
\end{pmatrix}\otimes
\begin{pmatrix}
    1 & 0\\
    0&\omega_N\\
\end{pmatrix}
=
\bigotimes_{j = 0}^{n - 1} P(\theta_j), \text{ where }
\theta_j =\frac{2 \pi}{N} \cdot 2^j.
\end{equation}
The implementation of $\Lambda$ is then as illustrated in \cref{figIncrementor}. We note that this structure is similar to the quantum adder proposed by~\citet{beauregard2003circuit} as a non-controlled version of Draper's adder circuit~\citep{draper2000addition}. The number of gates for this implementation is equal to the qubit number $n$.

We note that for $m\in \mathbb{Z}$, $Q^m$ has eigendecomposition $F^{-1}\Lambda^mF$. Further note that the powers of phase operations are phase operations with multiplied rotation angles, that is, $P(\theta)^m = P(m\theta)$. As $\Lambda$ can be decomposed into a tensor product of phase gates, we can simply implement $\Lambda^m$ with the same gate, but with the angles multiplied by $m$, as shown in \cref{figIncrementorPower}. Observe that the implemented quantum circuit is independent of $m$ and only depends on the depth of the QFT implementation.

\begin{figure}
    \centering
    \subfloat[\label{figIncrementor}]{
        \includegraphics[height=8em]{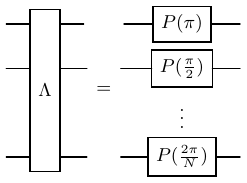}
    }
        \hfill
    \subfloat[\label{figIncrementorPower}]{
        \includegraphics[height=8em]{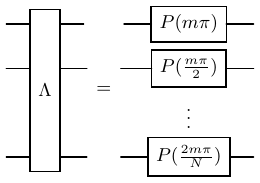}
    }
        \hfill
    \subfloat[\label{figCompleteRoutineHadamard}]{
        \includegraphics[height=8em]{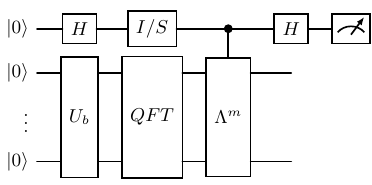}
    }
    \caption{Circuit implementation of quantum subroutines in our algorithm. The cyclic permutation matrix $Q$ can be diagonalized by the DFT matrix $F$ such that $Q = F^{-1}\Lambda F$. DFT can be implemented by QFT circuits. The diagonal matrix $\Lambda$ can be implemented as a tensor product of phase gates as shown in \cref{figIncrementor}. Powers of $\Lambda$ can be implemented with the same phase gates, by multiplying rotation angles as shown in \cref{figIncrementorPower}. To retrieve the real and imaginary parts of $\braket{b|Q^m|b}$, we use the Hadamard test (on the state $\mathsf{QFT}\ket{b}$), which requires a controlled version of $\Lambda^m$, i.e. controlled phase gates, as shown in \cref{figCompleteRoutineHadamard}. }
    \label{figAdder}
\end{figure}

Given the above implementation of powers of $Q$, to obtain the expectation value of $\braket{b|Q^m|b}$, we prepare a quantum state $\ket{\psi_b} = {\sf QFT} \ket{b}$. A Hadamard test is then performed on $\Lambda^m$ to calculate the expectation values with respect to $\ket{\psi_b}$ as described in \cref{appendixQroutines}, such that we obtain the values of $\re\braket{\psi_b|\Lambda^m|\psi_b}$ and $\im\braket{\psi_b|\Lambda^m|\psi_b}$, which are equivalent to $\re\braket{b|Q^m|b}$ and $\im\braket{b|Q^m|b}$. The full quantum circuit is shown in \cref{figCompleteRoutineHadamard}. 

\subsection{Number of total measurements.} As we conduct quantum measurements to obtain auxiliary system matrices $W$ and $r$, we do not gain the exact value of the entries in the matrices. Instead, what we obtain are estimations of the values whose accuracy depends on the number of measurements conducted. The total number of measurements required to achieve a close enough estimate of $\tilde x$ is stated in the following proposition:
\begin{proposition}[Number of measurements needed] \label{propMeasurement}
Given $K$-banded circulant matrix $C = \sum_{\ell=-K}^K c_\ell Q^\ell \in \mathcal M_{N}(\mathbb{C})$, hardware efficient implementation of $U_b$ as shown in \cref{assumptionHardware} and truncation threshold $T$. This defines auxiliary systems $W$ and $r$ in \cref{eqAuxiliary}. Let the sum of absolute values of the coefficients of decomposed $C$ be $B = \sum_{\ell=-K}^{K} |c_\ell|$. We can find $\tilde \alpha: \{\tilde\alpha_m \in \mathbb C, \forall m\in[-T.. T]\}$ for estimator $\tilde{x}(\alpha) = \sum_{m=-T}^T \alpha_{m}Q^m U_b\ket{0^n}$ such that the following is satisfied
\begin{equation*}
\|C\tilde x (\tilde \alpha) - \ket{b}\|^2 \le \min_{\alpha \in \mathbb C^{2T+1}} \|C\tilde x (\alpha) - \ket{b}\|^2+ \varepsilon
\end{equation*}
using $\mathcal O (B^4(K+T)^2\|W\|\|W^{-1}\|^2(1+\|W^{-1}r\|^2/\varepsilon)$ measurements via the Hadamard test.
\end{proposition}
\begin{proof}
We note that the number of different expectation values needed is $4K+4T+1 \in \mathcal O(K+T)$ due to the high overlap of values required. However, as the values of each entry are no longer independent, we cannot give a bound on the errors of each entry according to Proposition 12 of~\citep{huang2021term}. Suppose the additive error for each expectation value estimated produced by the Hadamard test is $\varepsilon_H$, which requires $\mathcal{O} (1/\varepsilon_H^2)$ repetitions. Then from \cref{propMatrixNorm}, the measurement errors are upper bounded by $\varepsilon_W = \|\hat{W}-W\|\le\tilde{\ca{O}}(\sqrt{(K+T)}\cdot B^2\varepsilon_H)$ and $\varepsilon_r = \|\hat{r}-r\|\le\ca{O}(\sqrt{(K+T)}\cdot B\varepsilon_H)$ from $\mathcal{O} ((K+T)/\varepsilon_H^2)$ measurements. Further, from \cref{propGuaranteesubspace}, we have $1/\varepsilon_W^2 > C \|W\| \|W^{-1}\|^2 \|z^*\|^2 / \varepsilon$ and $1/\varepsilon_r^2 > C \|W\| \|W^{-1}\|^2/ \varepsilon$. Plugging the bounds, we obtain $\frac{1}{\varepsilon_H^2} \in \tilde{\ca{O}} (B^4(K+T)\|W\|\|W^{-1}\|^2\|W^{-1}r\|^2/\varepsilon)$. Hence, the number of measurements required in total is then
\begin{equation}
\tilde{\ca{O}}\left(\frac{B^4(K+T)^2\|W\|\|W^{-1}\|^2\|W^{-1}r\|^2}{\varepsilon}\right).
\end{equation}
\end{proof}
Based on discussions given by \citet{huang2021term} on the upper bounds of $\|W\|$, $\|W^{-1}\|$, and $\|W^{-1}r\|$, we can see that the number of measurements in the proposition can be expressed as
\begin{equation}
\tilde{\ca{O}}\left(\frac{(K+T)^2B^4\kappa_C^2}{\varepsilon^5}\right).
\end{equation}
Further, in order to achieve an $\varepsilon$-close solution to the optimal $x^*$ where $x^* = \argmin_{x\in \mathbb C^N}\|Cx-b\|^2$, we note by \cref{propGuarantee} that $T\in O(K\cdot \kappa_C\log\frac{\kappa_C}{\varepsilon})$. Hence, the total number of measurements needed is 
\begin{equation}
   \tilde{\ca{O}}\left(\frac{K^2B^4\kappa_C^4}{\varepsilon^5}\right).
\end{equation}
To provide potential utility in terms of the number of runtimes, we would expect $K, B, \kappa_C \in \poly\log(N)$. Given that banded systems of interest are usually of constant size, e.g. tri-diagonal or penta-diagonal systems, we see that $K\in\mathcal{O}(1)$ as well as $B\in\mathcal{O}(K)\subseteq\mathcal{O}(1)$ in such cases. The restriction of well-conditioned systems on the conditional number $\kappa_C$ is a common assumption for efficient linear system solvers, both classical and quantum, as there is usually a polynomial dependency on $\kappa_C$ on linear system solvers. That being said, identifying well-conditioned linear systems that would bring quantum advantage remains a topic that requires further exploration in quantum algorithms related to linear system solving.

\section{Numerical results}\label{sectExp}
In this section, we consider the applications of solving banded circulant linear systems in finding solutions to partial differential equations (PDEs) with periodic boundaries. We consider the heat transfer problem in particular and obtain numerical results in solving the corresponding circulant linear systems.

\subsection{One-dimensional heat transfer problem.}
Solving heat transfer problems with quantum methods has been previously researched by \citet{guseynov2023depth}. They tackle the problem by applying QFT and using the variational quantum linear solver (VQLS)~\citep{bravo_prieto2020variational} for optimization in the Fourier domain. In addition, they also propose a CQS-Ansatz tree-based approach, which first converts the circulant matrix into the Fourier domain using QFT and then decomposes the Fourier representation into Pauli matrices. Their method requires gradient heuristics for potential speedups, as the inclusion of all Pauli matrices in the Ansatz tree to provide theoretical guarantees results in an exponential number of terms in relation to the qubit count. In comparison, we employ \cref{algoMain} for circulant linear systems with either the Hadamard test shown in \cref{figCompleteRoutineHadamard}, which uses QFT as diagonalization.

For simplicity,  we consider only the one-dimensional heat transfer problem. We denote the spatial variable by $x \in \mathbb R$ and the thermal diffusivity by $a \in \mathbb{C}$. The heat conduction equation is then given by
\begin{equation}
\label{heatEq}
a^2 \bm{\Delta} \Phi (x, t) - \frac{\partial \Phi(x, t)}{\partial t} = f(x, t),
\end{equation}
where $\bm{\Delta} = \nabla^2$ is the Laplacian operator, $\Phi(x, t)$ and $f(x, t)$ are the temperature and a known heat source at the grid position of $x$ at time $t$, respectively. Consider an initial time constraint and a $x_R \in \mathbb{R}$ periodic spatial boundary condition,
\begin{equation}
\label{heatBound}
    \Phi(x, 0) = \phi(x); \quad \Phi(x, t) = \Phi(x + x_R, t),
\end{equation}
where $\phi(x)$ is the initial distribution function. With the finite difference method~\citep{ozicsik2017finite}, we can obtain a numerical solution to the heat transfer problem by solving the corresponding linear systems of equations in the form of $Cx=\bm{b}$, where $x$ is the temperature distribution and $\bm{b}$ is a known vector~\citep{guseynov2023depth}. In this scenario, the matrix $C$ is given by
\begin{equation}
\label{heatLinear} C=
    \begin{pmatrix}
        -2 - \xi & 1 & 0 & \cdots & 0 & 1 \\ 
        1 & -2 - \xi & 1 & \ddots & \ddots & 0 \\ 
        0 & 1 &\ddots  & \ddots &\ddots  & \vdots \\
        \vdots & \ddots & \ddots & \ddots &  \ddots & 0\\
        0 & \ddots &\ddots  & \ddots & \ddots & 1 \\
        1 & 0 & \cdots & 0 & 1 & -2 - \xi \\
    \end{pmatrix},
\end{equation}
where $\xi \in \mathbb{R}$ is a grid parameter. As $C$ is a circulant matrix, it has the form of a linear combination of cyclic permutation matrices $Q$ and $Q^{-1}$, in addition to the identity matrix,
\begin{equation}
\label{heat_C}
C = (-2 - \xi) \mathbb{I} + Q + Q^{-1}.
\end{equation} 
The condition number of the matrix in \cref{heatLinear} is shown to be $\kappa_{C} = \frac{\xi + 4}{\xi}$ by \citet{guseynov2023depth}. Note that when $\xi \rightarrow 0$, the condition number $\kappa_{C} \rightarrow + \infty$. 

\begin{figure}
\centering
    \includegraphics[width=\onecolumnfigurewidth]{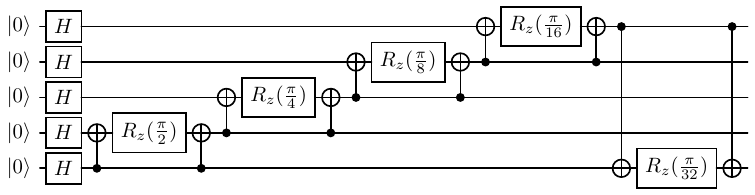}
    \caption{State preparation circuit of $\ket{b}$ modeled after the QAOA circuit~\citep{farhi2014quantum, hadfield2019quantum}.}
    \label{figQaoa}
\end{figure}

\begin{figure}
\centering
\subfloat[\label{figSimulationresults}]{
\includegraphics[width=\threecolumnfigurewidth,trim=85 25 85 60]{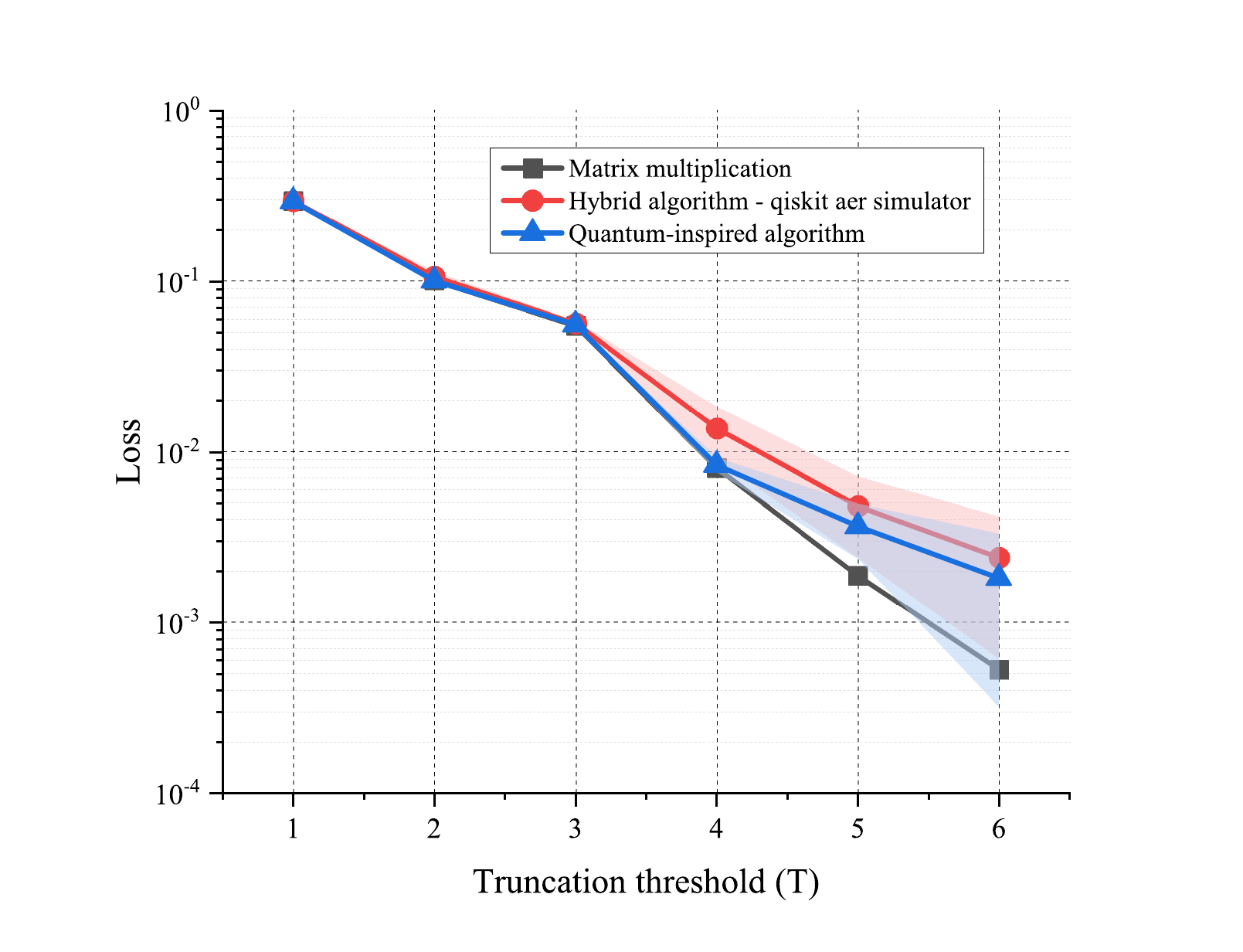}
}
\subfloat[\label{figQuantumresult}]{
\includegraphics[width=\threecolumnfigurewidth,trim=85 25 85 60]{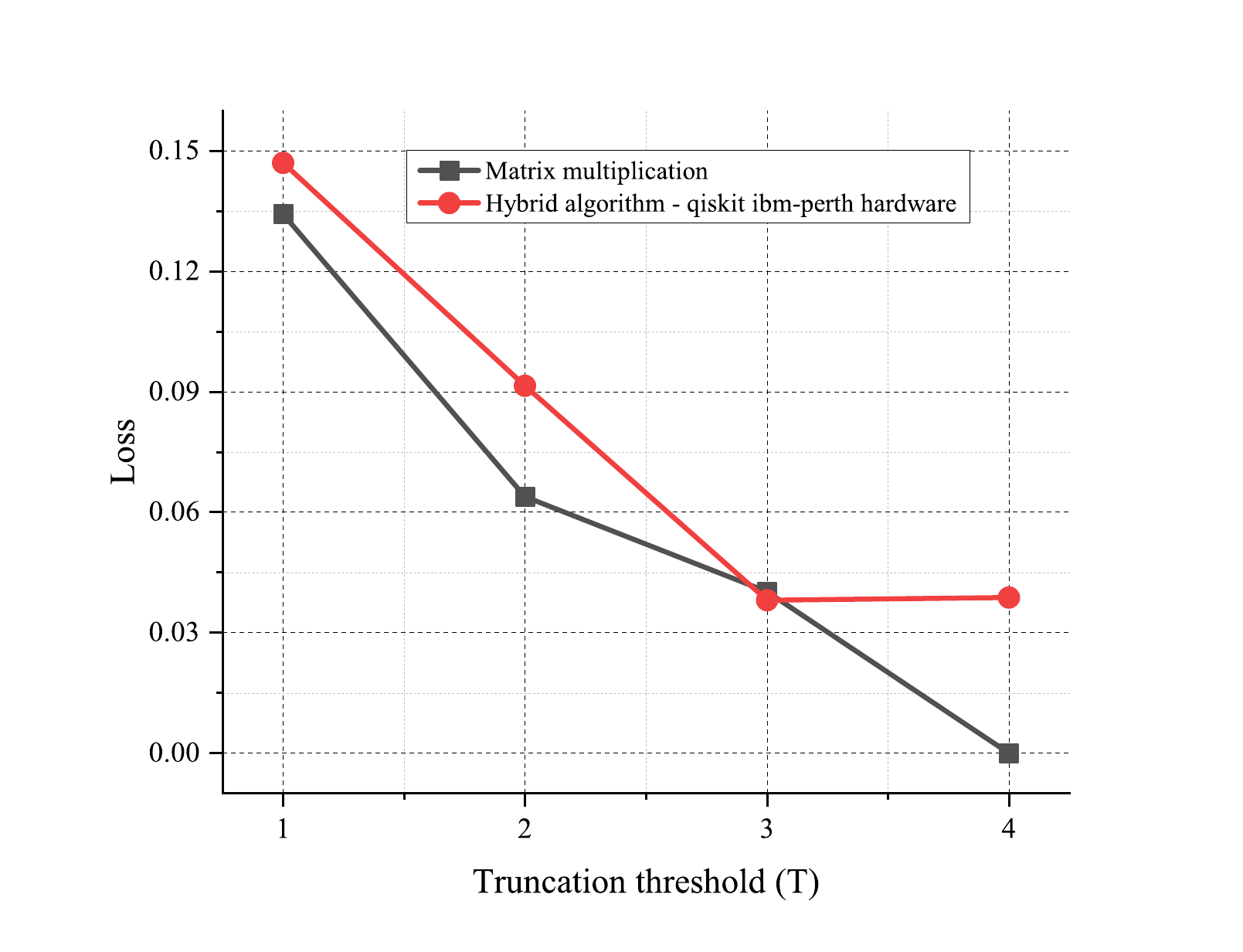}
}
\subfloat[\label{figCondition}]{
\includegraphics[width=\threecolumnfigurewidth,trim=85 25 85 60]{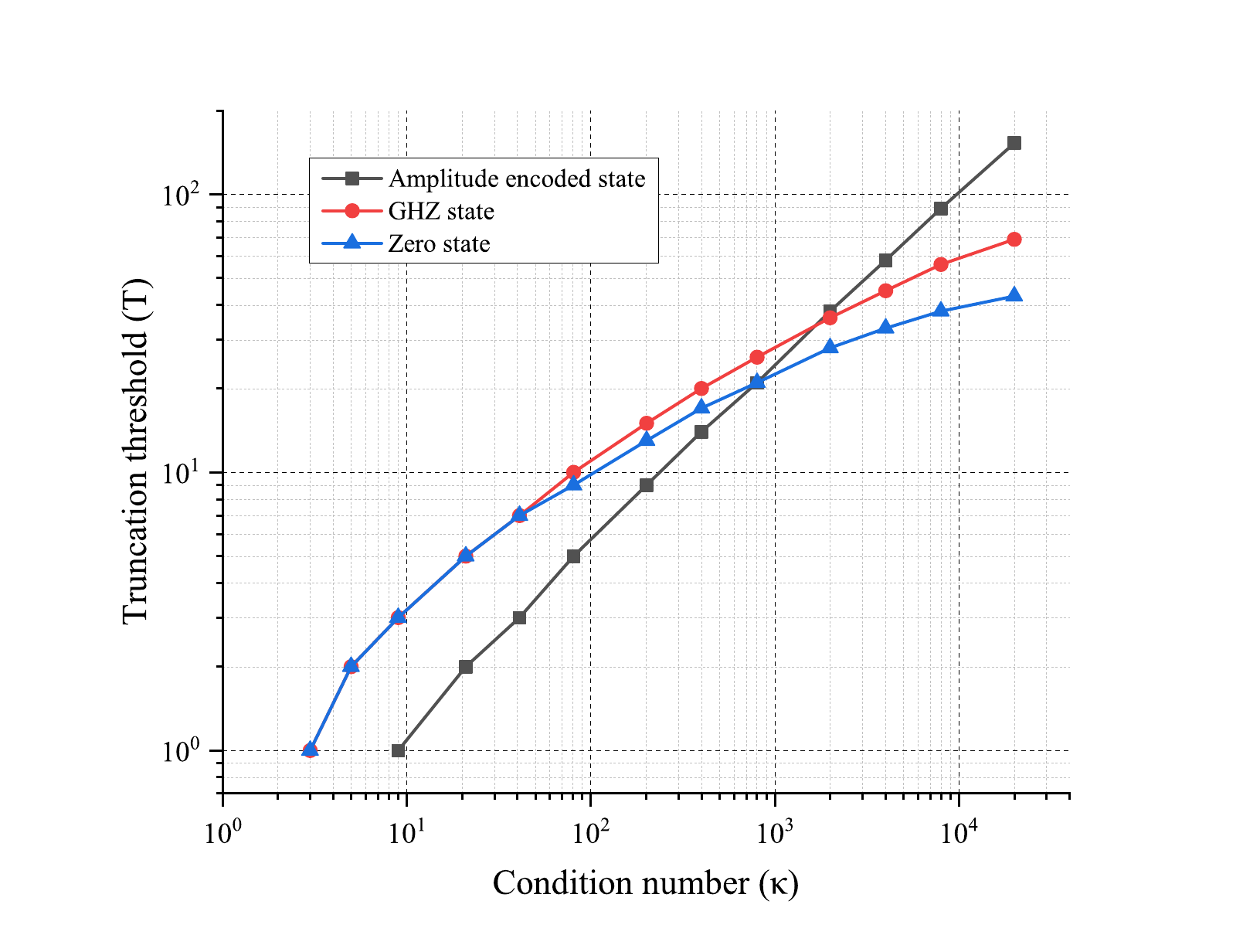}
}
\caption{Numerical results of the algorithm proposed in our paper for solving heat conduction linear systems. First, we show that simulation results of our hybrid and quantum-inspired algorithms can provide decent approximations to the results obtained by matrix multiplication in \cref{figSimulationresults}, with errors noticeable only in the logarithmic scale. Further, our hybrid algorithm is executable on IBM quantum hardware as shown in \cref{figQuantumresult}. Lastly, \cref{figCondition} shows that truncation thresholds fall within our provided upper bound in \cref{propGuarantee} in practice.}
\label{figNumerical}
\end{figure}

\subsection{Simulations and quantum hardware executions.} To demonstrate our algorithm in practice, we perform three numerical experiments, the results of which are summarized in \cref{figNumerical}. In this section, we also demonstrate the performance of an implementation of a quantum-inspired version of the algorithm, which replaces the usage of a Hadamard test with assumed sample and query access of the vector $\bm b$ and its permutations $Q^m\bm b$ per dequantization strategies discussed by \citet{tang2019quantum, tang2021quantum}, which we discuss further in \cref{sectSample}.

We first examine the performance of our algorithm on simulations of solving the heat transfer problem. We set the grid parameter $\xi=0.2$ and the system size $N=2^5$. The input state $\bm{b}$ is prepared by a single-layer QAOA circuit~\citep{farhi2014quantum, hadfield2019quantum} with the rotation parameters of $\{\frac{\pi}{2}, \frac{\pi}{4}, \frac{\pi}{8}, \frac{\pi}{16}, \frac{\pi}{32}\}$ as shown in \cref{figQaoa}. Note that we do not tune the parameters of the circuit, but merely use its structure to prepare our quantum state. We solve the corresponding circulant linear system of equations using three methods---direct matrix multiplication, simulation of the hybrid algorithm via the \texttt{qiskit-aer} simulator with $6\times10^4$ shots per Hadamard test, and simulation of the quantum-inspired algorithm with $6\times10^4$ accesses per sampling. The results are shown in \cref{figSimulationresults}. We note that our algorithms can provide a close estimate to that of the true value obtained by matrix multiplication, with errors noticeable only on the logarithmic scale.

Testing the performance of the hybrid algorithm on real quantum devices, we execute our hybrid algorithm on \texttt{ibm-perth} hardware. Due to the high consumption of quantum resources of the QFT subroutine on current quantum devices, we consider a smaller size of $N=2^3$ and set $\bm{b}$ as $\ket{0}^{\otimes 3}$. We set the shot number as $10^6$. We find that on small systems, our algorithms can still produce results close to the values obtained by the matrix multiplication, with the loss error not exceeding 0.05, as shown in \cref{figQuantumresult}.

Lastly, we compare the condition number with the truncation threshold under different input states $\ket{b}$. Setting the system size of $N=2^{10}$, we evaluate on three different types of $\bm{b}$, the zero state $\ket{0^{10}}$, the GHZ state $\ket{0^{10}}+\ket{1^{10}} / \sqrt{2}$ and the amplitude encoded state $\sum_{k=0}^{1023}k\ket {k} / \sqrt{\sum_{k=0}^{1023} k^2}$. We tune the parameter $\xi$ to adjust the condition number $\kappa_C$. Fixing a loss threshold $\upsilon=10^{-2}$, we record the minimum truncation threshold $T$ that achieves an MSE loss less than $\upsilon$ as shown in \cref{figCondition}. We note that from the slope of the log-log plot, we can obtain the exponent of $\kappa_C$ when the truncation threshold $T$ is modeled as a polynomial function of $\kappa_C$. We note that for the zero and GHZ state, the exponent seems to converge to a small value, suggesting the $T \in o(\kappa_C)$. Further, for the amplitude encoded state, we see that the slope is approximately $\frac{2}{3}$, indicating that $T \propto \kappa_c^{2/3}$. We note that our results all verify that the upper bound $T \in O(\kappa_C \log \kappa_C)$ provided by \cref{propGuarantee}; however, we have not found a vector $\bm b$ that saturates this upper bound, hence providing no empirical suggestions of the tightness of the proposition.

\section{Conclusion}
This work explores the CQS approach proposed by \citet{huang2021term} for solving circulant linear systems. As banded circulant systems can be expanded as a matrix polynomial of the cyclic permutation matrix $Q$, we do not require the iteratively growing Ansatz trees using heuristical gradient expansion strategies. We instead utilize the properties of circulant systems and only explore the direct sum of two Krylov subspaces $\mathcal K_T (Q, \bm b)$ and $\mathcal K_T (Q^{-1}, \bm b)$ to obtain the solution. 
We make use of efficient implementations of the powers of the cyclic permutation matrix $Q$ due to its eigendecomposition by QFT, as well as the availability of a tensor product decomposition of the diagonalized matrix $\Lambda$. The simplicity of the circulant problem gives rise to a classical algorithm that can be inferred from the theoretical guarantee of the linear-systems solver. 
We execute our algorithm in simulations as well as quantum hardware to solve partial differential equations produced by discretizing heat transfer problems, showing that these algorithms can indeed be utilized to produce approximate solutions that are close to the optimal solution obtained by exact algorithms that require a much higher runtime. Hence, this study shows the CQS method for analytically solvable cases and provides insights into more general applications of CQS when such analytical solutions are not available. 

While the work in our paper discusses the simple case of circulant systems, where there is an equivalence to a collapse of the Ansatz tree into a classically tractable set of operators, we believe that this study may provide insights into narrowing Ansatz tree exploration by studying cases where the structure of the problem may provide partial collapses of the Ansatz tree. For example, the closely related problem of PDEs with Dirichlet or Neumann boundary conditions can be viewed as a banded matrix system, and formulated as a sum of circulant systems and multiple rank-1 matrices. Such systems would thus produce deconstructions where the Ansatz tree constructed can be partially collapsed, enhancing the searchability of the Ansatz space, while preventing being classically tractable.

\begin{acknowledgments}
The authors would like to thank Bin Cheng for the valuable discussions and suggestions. This work is supported by the National Research Foundation, Singapore, and A*STAR under its CQT Bridging Grant and its Quantum Engineering Programme under grant NRF2021-QEP2-02-P05. XL and KK acknowledge support from the CQT Graduate Scholarship. KK acknowledges support from Leong Chuan Kwek, under project grant R-710-000-007-135. 
\end{acknowledgments}

\appendix
\counterwithin{theorem}{section}

\section{Additional proofs}\label{appendixGuarantee}
\subsection{Convergence guarantee}
In this section, we show the proof of the convergence guarantee of \cref{propGuarantee}.

\begin{proof}[Proof for \cref{propGuarantee}]
Our proof is divided into two parts; we first discuss the case where $C$ is Hermitian and then generalize the results to a non-Hermitian $C$.

{\it $C$ is Hermitian:} Discussing the case where $C$ is Hermitian, we first let $\hat C = C / \|C\|$. We find that the singular values of $\hat C$ are equal to the absolute value of the eigenvalues of $\hat C$, that is, $\{\sigma_k(\hat C)\}_{k\in[N]} = \{|\lambda_k(\hat C)|\}_{k\in[N]}$. Further notice that $\|\hat C\| = 1$ and $\kappa_{\hat C} = \kappa_C$. Noting that the eigenvalues of $ \hat C$ are all real, we can see that all eigenvalues of $\hat C$ fall within the domain $\mathcal{D}_{\kappa_C} = [-1, -\frac{1}{\kappa_C}] \cup [\frac{1}{\kappa_C}, 1]$, as limited by the largest and smallest singular values. Similar to the proof of proposition 4 of~\citep{huang2021term}, we find that by Lemma 14 of~\citep{childs2017quantum}, there exists a set of coefficients $j\in [0, j_0]$ such that $f(x) = \sum_{j=0}^{j_0} p_j x^j$ where $j_0 = \sqrt{\beta\log\frac{4\beta}{\nu}}$ and $\beta=\kappa_C^2\log\frac{\kappa_C}{\nu}$, such that $f(x)$ is 2$\nu$-close to $x^{-1}$ on domain $\mathcal{D}_{\kappa_C}$. We can obtain by diagonalization of powers of $\hat C$ that 
\begin{equation}
\left\|f(\hat C) - \hat C^{-1}\right\| = \max_{i \in [N]} \sigma_i\left(f(\hat C) - \hat C^{-1}\right) = \max_{i \in [N]} \left|\lambda_i(f(\hat C) - \hat C^{-1})\right| = \max_{i \in [N]} \left|f(\lambda_i(\hat C)) - \lambda_i(\hat C)^{-1}\right| \le 2 \nu.
\end{equation}
Let $g(x) = \sum_{j=0}^{j_0} \frac{p_j}{\|C\|^{j+1}} x^j$. Note that
\begin{equation}\label{eqepsilonClose}
    \|g(C) - C^{-1}\| = \left\|\sum_{j=0}^{j_0} \frac{p_j}{\|C\|^{j+1}} C^j - C^{-1}\right\| = \left\|\sum_{j=0}^{j_0} \frac{p_j}{\|C\|} \hat C^j - \frac{1}{\|C\|} \hat C^{-1}\right\| = \frac{\|f(\hat C) - \hat C^{-1}\|}{\|C\|} \le \frac{2\nu}{\|C\|}
\end{equation}
 Recall that $C = \sum_{\ell=-K}^K c_\ell Q^\ell$. Let $T = j_0K \in O(K\cdot\kappa_C\log\frac{\kappa_C}{\nu})$. We can find that there exists a set of parameters $\{\hat\alpha_i: i \in [-T.. T], \hat\alpha_i\in \mathbb{C}\}$ such that $g(C) = \sum_{j=0}^{j_0} \frac{p_j}{\|C\|^{j+1}} C^j = \sum_{m=-T}^{T} \hat{\alpha}_mQ^m$. We then find that
\begin{equation}
\min_{\{\alpha_m\}_{m\in [-T.. T]}} \left\|C\left(\sum_{m=-T}^{T} \alpha_mQ^m\right)\bm b - \bm b\right\|^2 \le \left\|C\left(\sum_{m=-T}^{T} \hat\alpha_mQ^m\right)\bm b - \bm b\right\|^2 = \|Cg(C)\bm b - \bm b\|^2
\end{equation}
We now observe the following:
\begin{align}
\|Cg(C)\bm b - \bm b\|^2 &= \|(C(g(C)-C^{-1})\bm b) + (CC^{-1}\bm b - \bm b)\|^2\\
&= \|C(g(C)-C^{-1})\bm b\|^2 + \|CC^{-1}\bm b - \bm b\|^2 + 2\re\braket{C(g(C)-C^{-1})\bm b, CC^{-1}\bm b - \bm b}
\end{align}
Given that 
\begin{align}
\braket{C(g(C)-C^{-1})\bm b, CC^{-1}\bm b - \bm b} &= \braket{(g(C)-C^{-1})\bm b, C^\dagger CC^{-1}\bm b - C^\dagger\bm b} \\
&= \braket{(g(C)-C^{-1})\bm b, C^\dagger\bm b - C^\dagger\bm b} = 0,
\end{align}
we have
\begin{equation}
\|Cg(C)\bm b - \bm b\|^2 = \|CC^{-1}\bm b - \bm b\|^2 + \|C(g(C)-C^{-1})\bm b\|^2.
\end{equation}
Observe that the first term is the MSE loss of the least-squares solution of $x=C^{-1}b$ and can hence be rewritten as $\min_{x\in \mathbb{C}^{N}}\|Cx-\bm b\|^2$.

The spectral norm of matrix $A$ can be defined by vector 2-norms such that
\begin{equation}
    \|A\| = \sup_{x\ne 0}\frac{\|Ax\|}{\|x\|},
\end{equation}
Hence, one can obtain an inequality
\begin{equation}
\|Ax\| \le \|A\|\|x\|.
\end{equation}
We then see that the second term $\|C(g(C)-C^{-1})\bm b\|^2$ can be upper bounded as follows:
\begin{equation}
\|C(g(C)-C^{-1})\bm b\|^2 \le \|C\|^2\|(g(C)-C^{-1})\bm b\|^2 \le \|C\|^2\|g(C)-C^{-1}\|^2\|\bm b\|^2 = \|C\|^2\frac{4\nu^2}{\|C\|^2}
\le 4\nu,
\end{equation}
with the last inequality using the fact that $0 < \nu \le 1$
Hence, we note that the solution obtained by our truncation threshold $T\in O(K\cdot\kappa_C\log\frac{\kappa_C}{\nu})$ has 
\begin{equation}
\min_{\{\alpha_m\}_{m\in [-T.. T]}} \left\|C\left(\sum_{m=-T}^{T} \alpha_mQ^m\right)\bm b - \bm b\right\|^2 \le \min_{x\in \mathbb{C}^{N}}\|Cx-\bm b\|^2 + 4\nu
\end{equation}

{\it $C$ is non-Hermitian:} Consider the following transformation and embedding matrices:
\begin{equation}
C' = \begin{pmatrix} O & C\\ C^\dagger & O \end{pmatrix}\quad {\bm b'} = \begin{pmatrix} \bm b \\ 0 \end{pmatrix} \quad x' = \begin{pmatrix} 0 \\ x \end{pmatrix},
\end{equation}
where $O$ is the zero matrix. Through the embedding, we can transform our linear system into solving $C'x'={\bm b'}$ where $C'$ is now Hermitian. We know that the eigenvalues of $C'$ correspond to the set $\{\pm \sigma_k(C)\}_{\forall k \in [N]}$~\citep{harrow2009quantum}. Hence, the singular values of $C'$ are equal to the set of singular values of $C$. Thus we have $\kappa_{C'} = \kappa_C$ and $\|C'\| = \|C\|$. Constructing $\hat C' = \frac{C'}{\|C\|}$, we see all eigenvalues of $\hat C'$ fall within the domain $\mathcal{D}_{\kappa_C}= [-1, -\frac{1}{\kappa_C}] \cup [\frac{1}{\kappa_C}, 1]$. 

By Lemma 14 of \citep{childs2017quantum} and \cref{eqepsilonClose}, we see that exists $g(x) = \sum_{j=0}^{j_0} \frac{p_j}{\|C\|^{j+1}} x^j$ such that 
\begin{equation}
\|g(C') - C'^{-1}\| \le \frac{2\nu}{\|C\|},
\end{equation}
To find a decomposition of $C'$, we first find the decomposition of $C^\dagger$. Using the fact that $Q$ is unitary, thus the adjoint/transpose of $Q^m$ is $Q^{-m}$:
\begin{equation}
    C^\dagger = \left(\sum_{\ell = -K}^K c_\ell Q^\ell\right)^\dagger = \sum_{\ell = -K}^K \bar c_\ell Q^{-\ell} = \sum_{\ell = -K}^K \bar c_{-\ell} Q^\ell,
\end{equation}
where $\bar c_\ell$ is the conjugate of $c_\ell$. We can then express $C'$ as follows 
\begin{equation}
C' = \begin{pmatrix} O & C\\ C^\dagger & O \end{pmatrix} = \begin{pmatrix} O & \sum\limits_{\ell = -K}^K c_{\ell} Q^\ell\\ \sum\limits_{\ell = -K}^K \bar c_{-\ell} Q^\ell & O \end{pmatrix} = \sum_{\ell=-K}^K \begin{pmatrix} 0 & c_\ell\\ \bar c_{-\ell} & 0 \end{pmatrix} \otimes Q^\ell
\end{equation}
One can find that for $T\in O(K\cdot\kappa\log\frac{\kappa}{\nu})$, there exists a set of parameters such that $\{\hat\alpha_m, \hat\beta_m, \hat\gamma_m, \hat\delta_m: m \in [-T..T], \hat\alpha_i, \hat\beta_m, \hat\gamma_m, \hat\delta_m \in \mathbb{C}\}$ such that we can write $g(C')$ in the following form:
\begin{align}
g(C') = \sum_{j=0}^{j_0} \frac{p_j}{\|C\|^{j+1}} C'^j &= \sum_{j=0}^{j_0} \frac{p_j}{\|C\|^{j+1}} \left(\sum_{\ell=-K}^{K}\begin{pmatrix} 0 & c_\ell\\ \bar c_{-\ell} & 0 \end{pmatrix} \otimes Q^\ell\right)^j\\
&= \sum_{m=-T}^T \begin{pmatrix} \hat\beta_i & \hat\delta_m\\ \hat\alpha_m & \hat\gamma_m \end{pmatrix} \otimes Q^m = \begin{pmatrix}
    \sum\limits_{m=-T}^T \hat\beta_m Q^m & \sum\limits_{m=-T}^T \hat\delta_m Q^m \\
    \sum\limits_{m=-T}^T \hat\alpha_m Q^m & \sum\limits_{m=-T}^T \hat\gamma_m Q^m
\end{pmatrix}
\end{align}
Furthermore, observe that the following equality is true:
\begin{align}
C'g(C'){\bm b'} &= \begin{pmatrix} O & C\\ C^\dagger & O \end{pmatrix}\begin{pmatrix} \sum\limits_{m=-T}^T \hat\beta_m Q^m & \sum\limits_{m=-T}^T \hat\delta_m Q^m \\ \sum\limits_{m=-T}^T \hat\alpha_m Q^m & \sum\limits_{m=-T}^T \hat\gamma_m Q^m \end{pmatrix}\begin{pmatrix}\bm b\\0\end{pmatrix} \\
&=  \begin{pmatrix} C\sum\limits_{m=-T}^T \hat\alpha_m Q^m & C\sum\limits_{m=-T}^T \hat\gamma_m Q^m \\ C^\dagger\sum\limits_{m=-T}^T \hat\beta_m Q^m & C^\dagger\sum\limits_{m=-T}^T \hat\delta_m Q^m \end{pmatrix}\begin{pmatrix}\bm b\\0\end{pmatrix} 
= \begin{pmatrix} C\left(\sum\limits_{m=-T}^{T} \hat\alpha_mQ^m\right)\bm b \\ C^\dagger\left(\sum\limits_{m=-T}^{T} \hat\beta_mQ^m\right)\bm b \end{pmatrix}
\end{align}
We now attempt to upper bound the MSE loss obtained from our estimator $\tilde x = \sum_{m=-T}^{T} \alpha_mQ^m$:
\begin{align}
\min_{\{\alpha_m\}_{m\in [-T.. T]}} \left\|C\left(\sum_{m=-T}^{T} \alpha_mQ^m\right)\bm b - \bm b\right\|^2 &\le \left\|C\left(\sum_{m=-T}^{T} \hat\alpha_mQ^m\right)\bm b - \bm b\right\|^2  \\
&= \left\|\begin{pmatrix} C\left(\sum\limits_{m=-T}^{T} \hat\alpha_mQ^m\right)\bm b \\ 0 \end{pmatrix} - \begin{pmatrix} b \\ 0 \end{pmatrix}\right\|^2\\
&\le \left\|\begin{pmatrix} C\left(\sum\limits_{m=-T}^{T} \hat\alpha_mQ^m\right)\bm b \\ C^\dagger\left(\sum\limits_{m=-T}^{T} \hat\beta_mQ^m\right)\bm b \end{pmatrix} - \begin{pmatrix} \bm b \\ 0 \end{pmatrix}\right\|^2\\
&\le \|C'g(C'){\bm b'}-{\bm b'}\|^2,
\end{align}
where the second-to-last inequality can be obtained by simply observing the fact that replacing zero terms in a vector with non-zero terms increases the value of the $\ell_2$ norm, and the last equality is obtained by replacing the term $\begin{pmatrix} C\left(\sum_{m=-T}^{T} \hat\alpha_mQ^m\right)\bm b && C^\dagger\left(\sum_{m=-T}^{T} \hat\beta_mQ^m\right)\bm b \end{pmatrix}^T$ with $C'g(C'){\bm b'}$.

Similar to the Hermitian case, we obtain the following:
\begin{equation}
\|C'g(C'){\bm b'} -{\bm b'}\|^2 = \|C'C'^{-1}{\bm b'}-{\bm b'}\|^2 + \|C'(g(C')-C'^{s-1}){\bm b'}\|^2.
\end{equation}
Noting that $C'^{-1} =\begin{pmatrix} O & C^{\dagger -1}\\ C^{-1} & O\end{pmatrix}$, we can recast the first term as follows:
\begin{align}
\|C'C'^{-1}{\bm b'}-{\bm b'}\|^2 &= \left\|\begin{pmatrix}O & C\\ C^\dagger & O\end{pmatrix}\begin{pmatrix} O & C^{\dagger -1}\\ C^{-1} & O\end{pmatrix}{\bm b'}-{\bm b'}\right\|^2 = \left\|\begin{pmatrix}CC^{-1} & O\\ O & C^\dagger C^{\dagger-1}\end{pmatrix}\begin{pmatrix}\bm b\\ 0 \end{pmatrix}-\begin{pmatrix}\bm b\\ 0 \end{pmatrix}\right\|^2\\
&=\left\|\begin{pmatrix}CC^{-1}\bm b \\ 0\end{pmatrix}-\begin{pmatrix}\bm b\\ 0 \end{pmatrix}\right\|^2 = \|CC^{-1}\bm b - \bm b\|^2 = \min_{x\in \mathbb{C}^{N}}\|Cx-\bm b\|^2
\end{align}
As $\|{\bm b'}\| = 1$, the second term is upper bounded as follows:
\begin{equation}
    \|C'(g(C')-C'^{-1}){\bm b'}\|^2 \le \|C'\|^2\|g(C')-C'^{-1}\|^2\|{\bm b'}\|^2 = \|C\|^2\frac{4\nu^2}{\|C\|^2} \le 4\nu^2 \le 4\nu
\end{equation}

Hence, we see that for the non-Hermitian case, the solution obtained by our truncation threshold $T\in\mathcal{O}(K\cdot\kappa_C\log\frac{\kappa_C}{\nu})$ is still $\nu$-close in terms of the MSE loss.
\end{proof}

\subsection{Measurement bounds}
In this section, we provide the lemmas and propositions deferred from the main text that are used to find the number of measurements required.
\begin{lemma}[Matrix Bernstein inequality; Corollary 6.1.2~\citep{tropp2015introduction}]\label{lemmaMatrixBernsteinInequality}
Consider a finite sequence $\left\{\hat{S}_k\right\}_k$ of independent random matrices with a common dimension $d_1 \times d_2$. Assume that each matrix has uniformly bounded deviation from its mean:
$$\forall k,\quad \left\|\hat{S}_k -\bar{S}_k\right\|\le L,$$ where $\bar{S}_k=\bb{E}\left[\hat{S}_k\right]$ denotes the mean value.
Let $\hat{Z} = \sum_k \hat{S}_k$ with its mean $\bar{Z}=\bb{E}\left[\hat{Z}\right]$ and let $\mathrm{MVar}\left[\hat{Z}\right]$ denote the matrix variance statistic given by:
\begin{align*}
\mathrm{MVar}\left[\hat{Z}\right] &= \max \left\{ \left\|\bb{E}\left[\left(\hat{Z}-\bar{Z}\right)\left(\hat{Z}-\bar{Z}\right)^*\right]\right\|, \left\|\bb{E}\left[\left(\hat{Z}-\bar{Z}\right)^*\left(\hat{Z}-\bar{Z}\right)\right]\right\| \right\}\\
&= \max \left\{ \left\|\sum_k \bb{E}\left[\left(\hat{S}_k - \bar{S_k}\right)\left(\hat{S}_k - \bar{S_k}\right)^*\right]\right\|, \left\|\sum_k \bb{E}\left[\left(\hat{S}_k - \bar{S_k}\right)^*\left(\hat{S}_k - \bar{S_k}\right)\right]\right\| \right\}.
\end{align*}
Then for all $t\ge 0$, 
$$\Pr\left[\left\|\hat{Z} - \bar{Z}\right\| \ge t\right] \le (d_1 + d_2) \cdot \exp \left( \frac{-t^2/2}{\mathrm{MVar}\left[\hat{Z}\right] + Lt/3} \right).$$
Further, we have the estimation with respect to the range of $t$ as given by \begin{align*}\Pr\left[\left\|\hat{Z} - \bar{Z}\right\| \ge t\right] \le \begin{cases} (d_1+d_2)\exp\left(-\frac{3t^2}{8\mathrm{MVar}\left[\hat{Z}\right]}\right), & t\le \frac{\mathrm{MVar}\left[\hat{Z}\right]}{L};  \\ (d_1+d_2)\exp\left(-\frac{3t^2}{8L}\right), & t> \frac{\mathrm{MVar}\left[\hat{Z}\right]}{L}.
\end{cases}
\end{align*} 

\end{lemma}
\begin{lemma}[Vector Bernstein inequality; Lemma 18~\citep{kohler2017subsampled}]\label{lemmaVectorBernsteinInequality}
Let $\hat{X}_1, \dots, \hat{X}_n$ be independent vector-valued random variables with common dimension $d$ and assume that each one is centered, uniformly bounded, and has a variance bounded from above:
$$\bar{X}_i=\bb{E}\left[\hat{X}_i\right]=0, \text{ and } \left\|\hat{X}_i\right\|_2\le \mu, \text{ as well as } \left\|\bb{E}\left[\hat{X}_i^2\right]\right\|\le\sigma^2$$
Let $\hat{Z}=\frac{1}{n}\sum_{i=1}^{n}\hat{X}_i.$ Then we have for $0<\varepsilon<\sigma^2/\mu$, $$\Pr\left[\left\|\hat{Z}\right\|\ge\varepsilon\right]\le\exp\left(-n\cdot\frac{\varepsilon^2}{8\sigma^2} + \frac{1}{4}\right)$$
\end{lemma}

We have the following proposition based on~\cref{lemmaMatrixBernsteinInequality} and~\cref{lemmaVectorBernsteinInequality} to support the above theorem to provide a tight upper bound on the measurement error of auxiliary system $W$ and $r$ in~\cref{eqAuxiliary}.
\begin{proposition}[Matrix norm bound]\label{propMatrixNorm}
Given $K$-banded circulant matrix $C = \sum_{\ell=-K}^K c_\ell Q^\ell \in \mathcal M_{N}(\mathbb{C})$, hardware efficient implementation of $U_b$ as shown in \cref{assumptionHardware} and truncation threshold $T$. Let $\hat{W}$ and $\hat{r}$ be the estimates of the auxiliary system $W$ and $r$ in \cref{eqAuxiliary}. Denote $\varepsilon_H$ as the additive errors caused by measurements, the spectral norms are bounded by
$$\left\|\hat{W}-W\right\|\le\ca{O}\left(\sqrt{(K+T)\log{\frac{T}{\delta}}}\cdot B^2\varepsilon_H\right),\quad \left\|\hat{r}-r\right\|\le\ca{O}\left(\sqrt{(K+T)\log{\frac{1}{\delta}}}\cdot  B\varepsilon_H\right),$$
with a probability at least $1-\delta$ using a total of $\ca{O}((K+T)/\varepsilon_H^2)$ measurements.
\end{proposition}
\begin{proof}
We first calculate the bound for $\left\|\hat{W}-W\right\|$. Note that $\hat{W}-W$ can be decomposed as a sum of tensor products such that 
\begin{equation}
\left\|\hat{W}-W\right\|=\left\|I\otimes \left(\re\left\{\hat{V}\right\} - \re\left\{V\right\}\right) -iY\otimes \left(\im\left\{\hat{V}\right\}-\im\left\{V\right\}\right)\right\|,
\end{equation}
where $I$ is the identity, $Y$ is the Pauli Y gate, and $\hat{V}$ is the estimates of $V$. By the triangle inequality, we have 
\begin{align}
\left\|\hat{W}-W\right\|&\le\left\|I\otimes \left(\re\left\{\hat{V}\right\} - \re\left\{V\right\}\right)\right\|+\left\|iY\otimes \left(\im\left\{\hat{V}\right\}-\im\left\{V\right\}\right)\right\|\\
&\le\left\|I\right\|\cdot\left\|\re\left\{\hat{V}\right\} - \re\left\{V\right\}\right\|+\left\|iY\right\|\cdot\left\|\im\{\hat{V}\}-\im\{V\}\right\|\\
&=\left\|\re\{\hat{V}\} - \re\{V\}\right\|+\left\|\im\{\hat{V}\}-\im\{V\}\right\|.
\end{align}
By \cref{eqVjk}, we know that $\hat{V}$ includes estimations of different inner products $\braket{\bm b,Q^{p}\bm b}$ with $p:=y-z+k-j$ for $y, z\in[-K, K]$ and $j, k\in[-T, T]$. The total number of different measurements is at most $4K+4T+1$ for estimating the real part or imaginary part. 

We denote $\re\left\{\hat{V}\right\}$ by $\hat{R}_V$ with mean $\bar{R}_V=\bb{E}\left[\hat{R}_V\right]$ and denote $\im\left\{\hat{V}\right\}$ by $\hat{I}_V$ with mean $\bar{I}_V=\bb{E}\left[\hat{I}_V\right]$. We first provide a bound for $\left\|\hat{R}_V - \bar{R}_V\right\|$. We obtain each measurement of $\re\left\{\braket{\bm b, Q^{y-z+k-j}\bm b}\right\}$ by setting an estimator $\hat{o}^{(p)}$ with mean $\bar{o}^{(p)}=\bb{E}\left[\hat{o}^{(p)}\right]$. Then both $\hat{R}_V$ and $\bar{R}_V$ are written as summations of $4K+4T+1$ independent terms given by \begin{align}\hat{R}_V=\sum_{p=-2K-2T}^{2K+2T}\hat{R}_V^{(p)}=\sum_{p=-2K-2T}^{2K+2T}R_V^{(p)} \hat{o}^{(p)}, \quad\bar{R}_V=\sum_{p=-2K-2T}^{2K+2T}\bar{R}_V^{(p)}=\sum_{p=-2K-2T}^{2K+2T}R_V^{(p)} \bar{o}^{(p)}.
\end{align}
One can then observe that $R_V^{(p)}$ is a Toeplitz matrix in which each diagonal from left to right is constant. According to different value of $p$, we categorize the matrix structures of $R_V^{(p)}$ into five types given by 
\begin{align}\label{eqRp}R_V^{(p)}=
\begin{cases}
     BT1_V^{(p)},& p\in[-2K-2T.. 2K-2T]; \\
     BT2_V^{(p)},& p\in[2K-2T+1.. -2K-1]; \\
     BT3_V^{(p)},& p\in[-2K.. 2K];  \\
     BT4_V^{(p)},& p\in[2K+1.. -2K+2T-1]; \\
     BT5_V^{(p)},& p\in[-2K+2T.. 2K+2T].
\end{cases}
\end{align}
For an explicit illustration, we depict the above five types of banded Toeplitz (BT) matrices with elements from the coefficient set $\{a_{-2K},\dots, a_{2K}\}$ with the structures shown below.
\begin{align*}
     &\resizebox{0.48\textwidth}{!}{$BT1_V^{(p)}=\begin{bmatrix}
     0 & \cdots & \cdots & \cdots& \cdots& \cdots& \cdots& \cdots & 0\\ 
     \vdots & & & & & & & & \vdots \\
     \vdots & & & & & & & & \vdots \\
     \vdots & & & & & & & & \vdots \\
     \vdots & & & & & & & & \vdots \\
     0 & & & & & & &  & \vdots\\
     a_{-2K} &\ddots &  & &  & && & \vdots \\
     \vdots & \ddots & \ddots &  & &  && & \vdots \\
     a_{2T+p} & \cdots & a_{-2K} & 0 & \cdots & \cdots& \cdots& \cdots & 0 \end{bmatrix}$},
     \resizebox{0.48\textwidth}{!}{$BT2_V^{(p)}=\begin{bmatrix}
     0 & \cdots & \cdots & \cdots& \cdots& \cdots& \cdots& \cdots & 0\\
     \vdots & & & & & & & & \vdots \\
     0& & & & & & & & \vdots \\
     a_{-2K}  &\ddots & & & & & & & \vdots \\
     \vdots &\ddots &\ddots & & & & & & \vdots \\
     a_{2K} &\ddots &\ddots &\ddots & & & &  & \vdots\\
     0 &\ddots &\ddots  &\ddots &\ddots  & && & \vdots \\
     \vdots & \ddots & \ddots &\ddots  &\ddots &\ddots  && & \vdots \\
     0 & \cdots & 0& a_{2K} & \cdots & a_{-2K}& 0& \cdots & 0 \end{bmatrix}$}, \\
     &\resizebox{0.48\textwidth}{!}{$BT3_V^{(p)}= \begin{bmatrix}
     a_{p} & \cdots & a_{-2K} & 0& \cdots& \cdots& \cdots& \cdots & 0\\ 
     \vdots &\ddots &\ddots &\ddots &\ddots & & & & \vdots \\
     a_{2K}&\ddots &\ddots &\ddots &\ddots &\ddots & & & \vdots \\
     0  &\ddots &\ddots &\ddots &\ddots &\ddots & \ddots && \vdots \\
     \vdots &\ddots &\ddots &\ddots &\ddots &\ddots &\ddots &\ddots & \vdots \\
     \vdots & &\ddots &\ddots &\ddots &\ddots &\ddots &\ddots  & 0\\
     \vdots &&  &\ddots &\ddots  &\ddots &\ddots&\ddots & a_{-2K} \\
     \vdots & &  & &\ddots &\ddots  &\ddots&\ddots & \vdots \\
     0 & \cdots & \cdots&\cdots & \cdots & 0& a_{2K}& \cdots & a_{p} \end{bmatrix}$},
     \resizebox{0.48\textwidth}{!}{$BT4_V^{(p)}=\begin{bmatrix}
     0 & \cdots & 0 & a_{2K}& \cdots& a_{-2K}& 0& \cdots & 0\\ 
     \vdots & & &\ddots &\ddots &\ddots & \ddots&\ddots & \vdots \\
     \vdots& & & &\ddots &\ddots &\ddots &\ddots & 0 \\
     \vdots  & & & &&\ddots & \ddots &\ddots& a_{-2K} \\
     \vdots & & & & & &\ddots &\ddots & \vdots \\
     \vdots & && & & & &\ddots  & a_{2K}\\
     \vdots &&  & &  & && & 0 \\
     \vdots & &  & & &  && & \vdots \\
     0 & \cdots & \cdots&\cdots & \cdots & \cdots& \cdots& \cdots &0 \end{bmatrix}$},\\
     &\resizebox{0.48\textwidth}{!}{$BT5_V^{(p)}=\begin{bmatrix}
     0 & \cdots & \cdots &\cdots& \cdots& 0& a_{2K}& \cdots & a_{-2T+p}\\ 
     \vdots & & & & & &\ddots &\ddots & \vdots \\
     \vdots&&&&&&&\ddots & a_{2K} \\
     \vdots &&&&&&&& 0\\
     \vdots &&&&&&&& \vdots \\
     \vdots &&&&&&&& \vdots \\
     \vdots &&&&&&&& \vdots \\
     \vdots &&&&&&&& \vdots \\
     0 & \cdots & \cdots&\cdots & \cdots & \cdots& \cdots& \cdots &0 \end{bmatrix}$}.
\end{align*}
We note that the 1-norm and $\infty$-norm of matrix $R_V^{(p)}$ are given by \begin{align}
    \left\|R_V^{(p)}\right\|_1, \left\|R_V^{(p)}\right\|_{\infty}\le \sum_{m=-2K}^{2K} \left|a_{m}\right|\le B^2.
\end{align}
Hence the spectral norm of matrix $R_V^{(p)}$ is also upper bounded by \begin{align}
    \left\|R_V^{(p)}\right\| \le \sqrt{\left\|R_V^{(p)}\right\|_1 \cdot \left\|R_V^{(p)}\right\|_{\infty}}\le B^2.
\end{align}
By Hadamard tests, measurement results are collected to obtain a good estimation of $\hat{o}^{(p)}$ that
$\left|\hat{o}^{(p)} - \bar{o}^{(p)}\right|\le \varepsilon_H$,  using $\mathcal{O}(\frac{1}{\varepsilon_H^2}\log \frac{K+T}{\delta_H})$ samples with a success probability at least $1-\delta_H$ for all $p$. Hence the error bound of each random variable $\hat{R}_V^{(p)}$ is given by \begin{align}\label{eqRandomVariableBound}
    \left\|\hat{R}_V^{(p)}-\bar{R}_V^{(p)}\right\|=\left\|R_V^{(p)}\right\|\left|\hat{o}^{(p)}-\bar{o}^{(p)}\right|\leq B^2 \varepsilon_H =: L.
\end{align}
By the triangle inequality, we observe that 
\begin{align}\label{eqSummationOfRandomVariablesBound}
\left\|\hat{R}_V-\bar{R}_V\right\|\leq\sum_{p=-2K-2T}^{2K+2T}\left\|\hat{R}_V^{(p)}-\bar{R}_V^{(p)}\right\|\le (4K+4T+1)B^2\varepsilon_H.
\end{align}
The matrix variance statistic $\textrm{MVar}\left[\hat{R}_V\right]$ is also upper bounded by the triangle inequality that
\begin{align}\label{eqMVarBound}
\textrm{MVar}\left[\hat{R}_V\right]\le \sum_{p=-2K-2T}^{2K+2T} \left\|\hat{R}_V^{(p)}-\bar{R}_V^{(p)}\right\|^2\le (4K+4T+1)B^4\varepsilon_H^2.
\end{align}
For $t\le \left\|\hat{R}_V-\bar{R}_V\right\|\leq (4K+4T+1) B^2\varepsilon_H=\textrm{MVar}\left[\hat{Z}\right]/L$, we leverage matrix Bernstein inequality in \cref{lemmaMatrixBernsteinInequality} to get that, for all $t\ge 0$, we have
\begin{align}
    \Pr\left[\left\|\hat{R}_V-\bar{R}_V\right\|\ge t\right] \le (4T+2)\exp\left[-\frac{3t^2}{8(4K+4T+1)B^4\varepsilon_H^2}\right].
\end{align}
We then restrict the failure probability by a small value $\delta_{R_V} \in (0, \frac{1}{2})$ and get \begin{align}
    (4T+2)\exp\left[-\frac{3t^2}{8(4K+4T+1)B^4\varepsilon_H^2}\right]\leq \delta_{R_V}.
\end{align}
We can then find that $t=\ca{O}(\sqrt{(K+T)\log{(T/\delta_{R_V})}}B^2\varepsilon_H)$. Therefore, we can see the bound holds that \begin{align} \left\|\hat{R}_V-\bar{R}_V\right\|\le \ca{O}(\sqrt{(K+T)\log{(T/\delta_{R_V})}}B^2\varepsilon_H)\end{align} with a probability at least $1-\delta_{R_V}$. The same argument applies when estimating the imaginary part, and hence we have \begin{align} \left\|\hat{I}_V-\bar{I}_V\right\|\le \ca{O}(\sqrt{(K+T)\log{(T/\delta_{I_V})}}B^2\varepsilon_H)\end{align} with a probability at least $1-\delta_{I_V}$. By setting $\delta_{R_V} = \delta_{I_V} = \delta/2$, the spectral norm of $\hat{W}-W$ is given by \begin{align}\|\hat{W}-W\|\le \left\|\hat{R}_V-\bar{R}_V\right\| + \left\|\hat{I}_V-\bar{I}_V\right\| \le \ca{O}(\sqrt{(K+T)\log{(T/\delta)}}B^2\varepsilon_H),\end{align} with a probability at least $1-\delta$ by matrix Bernstein inequality. 

For the bound on $\left\|\hat{r}-r\right\|$, we utilize the vector Bernstein inequality in \cref{lemmaVectorBernsteinInequality}. Given that $\left\|\hat{r}-r\right\|\le \left\|\re\left\{\hat{r}\right\}-\re\{r\}\right\| + \left\|\im\left\{\hat{r}\right\}-\im\{r\}\right\|$, we again denote $\re\left\{\hat{r}\right\}$ by $\hat{R}_r$ with mean $\bar{R}_r=\bb{E}\left[\hat{R}_r\right]$ and denote $\im\left\{\hat{r}\right\}$ by $\hat{I}_r$ with mean $\bar{I}_r=\bb{E}\left[\hat{I}_r\right]$. Given that \begin{align}\left\|\hat{R}_r-\bar{R}_r\right\|=\left\|\sum_{p=-K-T}^{ K+T}R_r^{(p)}\hat{o}^{(p)}\right\|,\end{align} where $R_r^{(p)}$ has at most $2K+1$ nonzero entries and $\hat{o}^{(p)}$ is the estimator corresponding to the real part of inner product estimation of $\braket{\bm b, Q^{p}\bm b}$ with $p:=y+j$ for $y\in[-K, K]$ and $j\in[-T, T]$. Given that $\left\|R_r^{(p)}\right\| \le \left\|R_r^{(p)}\right\|_1\le B$, we have \begin{align}\left\|\hat{R}_r-\bar{R}_r\right\|\leq (2K+2T+1) B \varepsilon_H.\end{align} We then utilize the vector Bernstein inequality in \cref{lemmaVectorBernsteinInequality} and follow a similar proof as shown in bounding the matrix norm. Hence we can then obtain a tighter bound of \begin{align}\|\hat{r}-r\|\le\ca{O}(\sqrt{(K+T)\log{(1/\delta_r)}}B\varepsilon_H),\end{align} with a probability at least $1-\delta_r$. Taking $\delta_r =\delta$ finishes the proof.

\end{proof}

\begin{proposition}[Proposition 12~\citep{huang2021term}]
\label{propGuaranteesubspace}
Consider a quadratic function $L(z) = z^T W z - 2 r^T z $, where $W \in R^{m\times m}$ is positive definite and $z \in \mathbb{R}^m$.
Let $z^* = {\rm arg} \min_{z \in \mathbb R^m} L(z)$.
Let  $\hat{W}$ and $\hat{r}$ be estimates of $W$ and $r$ where $\|\hat W- W\| = \varepsilon_W$ and $\|\hat r- r\| = \varepsilon_r$.
The solution $\hat z$ of the quadratic optimization $\hat{L}(z) = z^T \hat{W} z - 2 \hat{r}^T z$, satisfies
\begin{equation*}
 L(\hat z) - L(z^*) \le \varepsilon,
\end{equation*}
if $\max(1/\varepsilon_W^2, 1/\varepsilon_r^2) > C \|W\| \|W^{-1}\|^2 (1+\|z^*\|)^2 / \varepsilon$.
\end{proposition}

\begin{proof}

Let $z^* = W^{-1} r$ be the solution to the original problem. The solution $\hat z = \hat{W}^{-1} \hat{r}$ minimizes  $\hat{L}(z) = z^T \hat{W} z - 2 \hat{r}^T z$.
Thus we have
$\hat{W} \hat z - W \hat z + W \hat z - W z^* = \hat{r} - r.$
This gives $(\hat z - z^*) = W^{-1} (\hat{r} - r - (\hat{W} - W) \hat z)$. Hence $\|\hat z - z^*\|\le \|W^{-1}\|(\|\hat r- r\| + \|\hat W- W\| \|\hat z\|) \le \|W^{-1}\| \|\hat r- r\| + \|W^{-1}\| \|\hat W- W\| \|z^*\| + \|W^{-1}\| \|\hat W- W\| \|\hat z - z^*\|$.
This gives
\begin{align}\left\|\hat z - z^*\right\| \le \frac{\left\|W^{-1}\right\|\left(\left\|\hat r- r\right\| + \left\|\hat W- W\right\| \left\|z^*\right\|\right)}{1 - \left\|W^{-1}\right\| \left\|\hat{W}-W\right\|} \le \sqrt{\varepsilon / \|W\|}.\end{align}
The last inequality requires setting $1/\varepsilon_W^2 > C \|W\| \|W^{-1}\|^2 \|z^*\|^2 / \varepsilon$ and $1/\varepsilon_r^2 > C \|W\| \|W^{-1}\|^2/ \varepsilon$, with a constant $C$, such that $1 - \|W^{-1}\| \|\hat W- W\| \ge 1/2$ , $\|W^{-1}\|\|\hat r- r\| \le \frac{1}{4} \sqrt{\varepsilon / \|W\|}$, and $ \|W^{-1}\|\|\hat W- W\| \|z^*\| \le \frac{1}{4} \sqrt{\varepsilon / \|W\|}$.
As the gradient of $L(z)$ is zero at $z^*$, we have $L(\hat z) - L(z^*) \le \|W\| \left\|\hat z - z^*\right\|^2 \le \varepsilon$.
\end{proof}

\section{Quantum subroutines}\label{appendixQroutines}
In the following section, we provide brief reviews of quantum subroutines used in our paper.

\subsection{Hadamard test.} Hadamard tests~\citep{cleve1998quantum} are used as a method to obtain a random variable whose expectation value is the real (or imaginary) part of $\braket{\psi|U|\psi}$, where $\ket \psi$ is a quantum state and $U$ is unitary.  Given an additional ancilla qubit, one can prepare the controlled version of $U$ and run the quantum circuit given by \cref{figRealHadamard} for finding the real part and \cref{figImagHadamard} for finding the imaginary part. To compute the expectation value, obtain the probability of measuring 1 minus the probability of measuring 0. Invert the sign when obtaining the imaginary part. Note that the expected values produced by Hadamard test have additive accuracy $\varepsilon$ with failure probability at most $\delta$ using $\mathcal{O}\left(\frac{1}{\varepsilon^2} {\rm log} (\frac{1}{\delta})\right)$ measurements, which can be shown using Hoeffing's inequality.

\begin{figure}
    \centering
    \subfloat[\label{figRealHadamard}]{
        \includegraphics[height=4em]{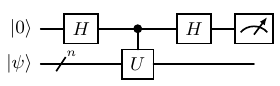}
    }
    \subfloat[\label{figImagHadamard}]{
        \includegraphics[height=4em]{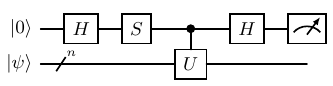}
    }

    \caption{Circuit implementation of the Hadamard test. \cref{figRealHadamard} is the circuit for retrieving $\re\braket{\psi|U|\psi}$, while \cref{figImagHadamard} is the circuit for retrieving $\im\braket{\psi|U|\psi}$.}
    \label{figHadamardTest}
\end{figure}

\subsection{Quantum Fourier transform.} 
Given $N=2^n$, the quantum Fourier transform~\citep{coppersmith1994approximate} performs the following transformation on a quantum state:
\begin{equation}
\ket b = \sum_{x=0}^{N-1} b_x \ket x \longrightarrow \underbrace{\frac{1}{\sqrt N} \sum_{x=0}^{N-1} b_x \omega_N^{xy}\ket y}_{F_N\ket b}.
\end{equation}
Given the following gate:
\begin{equation}
    R_k = \begin{pmatrix} 1 & 0\\ 0 & \omega_N^k \end{pmatrix},
\end{equation}
QFT can be implemented as \cref{figQFT}.

\begin{figure}
\centering
\includegraphics[width=0.98\textwidth]{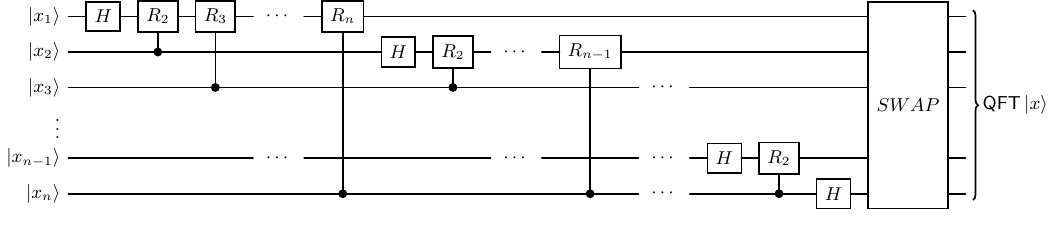}
    \caption{Implementation of the quantum Fourier transform, where the final SWAP gate reverses the order of qubits.}
    \label{figQFT}
\end{figure}
The quantum Fourier transform as proposed by \citet{coppersmith1994approximate} is implemented as a unitary acting on $n$ qubits, which makes it an operator of size $2^n \times 2^n$. One might be interested in discrete Fourier transforming a vector of $N$-amplitudes, where $2^{n-1} < N < 2^n$ for some $n$. This could be achieved via generalized QFT (first conceived of as an approximation by \citet{kitaev1995quantum} and made exact by \citet{mosca2004exact}), which makes use of quantum phase estimation (QPE) proposed in the former paper, which itself incorporates the `standard' QFT by \citet{coppersmith1994approximate}.

\section{Quantum-inspired algorithm}~\label{sectSample}
We now elaborate on the implementation of \cref{algoMain} using purely classical computers that use techniques used by \citet{tang2019quantum} to estimate inner products. Although this is not the most efficient classical algorithm, it serves as a direct comparison to our quantum algorithm using classical combinations of quantum states. To use these techniques, we require sample and query access to the vector $\bm b$.
\begin{definition}[Sample and query access~\citep{tang2019quantum, tang2021quantum}]
For vector $x\in \mathbb{C}^N$, let $\mathcal{D}_x$ be the distribution over $[N]$ such that a sample $X\sim\mathcal{D}_x$ satisfies $\Pr(X=k) = \frac{|x_k|^2}{\|x\|^2}$. A sample access of $x$ returns $k \in [N]$ over $\mathcal{D}_x$. A query access of $x$ with supplied index $k$ returns the $k$-th element of $x$.
\end{definition}
We then make the following assumption:
\begin{assumption}\label{assumptionSample}
Assume that there exists a classical data structure $\mathcal{S}_{\bm b}$ containing sufficient information of vector $\bm b$ such that command {\sc Sample($\mathcal{S}_{\bm b}$)} achieves sample access to $\bm b$ and command {\sc Query($\mathcal{S}_{\bm b}$, $i$)} achieves query access to the $i$-th element of $\bm b$. Further, assume that the input of vector $\bm b$ to \cref{algoMain} is given in the form of $\mathcal{S}_{\bm b}$.
\end{assumption}
A binary tree implementation of such a data structure has been shown by~\citet{chia2018quantuminspired}. The leaf nodes of the tree store the square of the absolute value of each entry, in addition to its original value, while the internal nodes store the sum of their children. As a result, the root node stores the square of the $\ell_2$ norm of the vector. Each sample and query access then requires $\mathcal O(\log N)$ time. 

\begin{figure}
\begin{algorithm}[H]
\caption{Classical subroutine for inner product estimation with sample and query access}
\label{algoAccess}
\DontPrintSemicolon
\KwIn{Classical data structure $\mathcal{S}_{\bm b}$ as shown in \cref{assumptionSample}, Power of cyclic perturbation matrix $m$, Additive error $\varepsilon$, Failure rate $\delta$}
\KwOut{$\varepsilon$-close estimation of $\braket{\bm b, Q^m \bm b}$ with success rate $1-\delta$}
\For{$i \gets 0$ to $6\log(2/\delta)$}{
    $\eta_i \gets 0$\;
    \For{$j \gets 0$ to $9/\varepsilon^2$}{
        $s \gets \textsc{Sample}(\mathcal{S}_{\bm b})$\;
        $q \gets \textsc{Query}(\mathcal{S}_{\bm b}, s-m \pmod N)$\;
        $r \gets \textsc{Query}(\mathcal{S}_{\bm b}, s)$\;
        $\eta_i \gets \eta_i + q/r$\;
    }
    $\eta_i \gets \eta_i$\;
}
\Return \textsc{Median}($\eta$)
\end{algorithm}
\end{figure}

\begin{figure}
\includegraphics[width=\onecolumnfigurewidth]{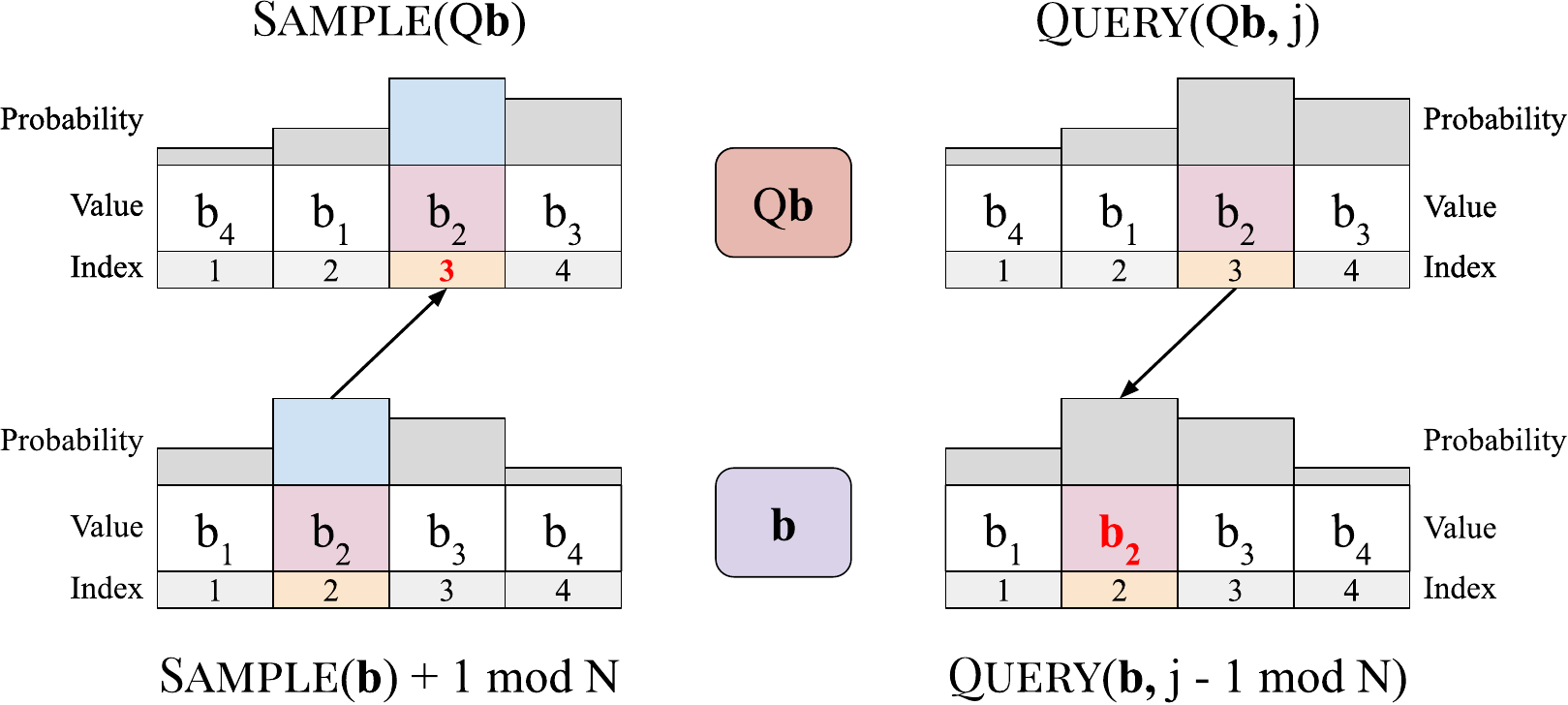}
\caption{We illustrate how sample and query access to vector $\bm b$ provides the same access to $Q\bm b$. Sample access to $Q\bm b$ is equivalent to shifting the output of sample access to $\bm b$, while query access to $Q\bm b$ is equivalent to shifting the input to query access to $\bm b$. Note that sample access to $Q\bm b$ is not needed for \cref{algoAccess}.}
\label{figSample}
\end{figure}

Although the hybrid algorithm partially implemented on quantum hardware decomposes $Q$ into a QFT operation and phase shift operations, the quantum-inspired version does not perform any Fourier-related operations due to its high classical time complexity. Instead, we use the fact that the application of $Q$ on vector $\bm b$ performs a cyclic permutation of vector $\bm b$. Observe that the query access to $Q^m\bm b$ retrieves the elements of $b$ shifted $m$ times to the right and sample access to $Q^m\bm b$ outputs an index that corresponds to the index $m$ times to the right of the original vector $\bm b$, as illustrated in \cref{figSample}.

We estimate the value of $\braket{\bm b, Q^m\bm b}$ by \cref{algoAccess}, which is modified from Algorithm 1 of~\citep{tang2021quantum}, and whose performance guarantee of requiring $\mathcal O (\frac{1}{\varepsilon^2}\log\frac{1}{\delta})$ sample and query accesses to execute is ensured by Proposition 3 of the same paper. Such a process can effectively replace the use of Hadamard tests, allowing us to perform Steps 1 to 4 of \cref{algoMain} on classical computers. 

Further, the classical sketch of the output vectors $\bm u_m$ can be constructed by using an index shifting wrapper function $\mathcal F_m$ that wraps data structure $\mathcal{S}_{\bm b}$ such that the command \textsc{Sample}($\mathcal F_m(\mathcal{S}_{\bm b})$) returns \textsc{Sample}($\mathcal{S}_{\bm b}$)$+m \pmod N$ and command \textsc{Query}($\mathcal F_m(\mathcal{S}_{\bm b}, i)$) returns \textsc{Query}($\mathcal{S}_{\bm b}, i-m \pmod N$). We can see that the sample and query accesses to $\mathcal S_{\bm u_m}$ is equivalent to $\mathcal{F}_m(\mathcal{S}_{\bm b})$, and $\mathcal{F}_m(\mathcal{S}_{\bm b})$ can be used to represent the sketch for $\mathcal{S}(\bm u_m)$. 

With the above, we have produced a quantum-inspired classical algorithm with the same performance guarantee. Note that the total number of query and sample accesses required to execute \cref{algoMain} matches the upper bound of quantum measurements required as supplied in \cref{propMeasurement}, which we restate in the current setting:

\begin{lemma}[Number of sample/query accesses needed] \label{propAccess}
Given $K$-banded circulant matrix $C = \sum_{\ell=-K}^K c_\ell Q^\ell \in \mathcal M_{N}(\mathbb{C})$, classical data structure $\mathcal{S}_{\bm b}$ as shown in \cref{assumptionSample} and truncation threshold $T$. This defines the auxiliary systems $W$ and $r$ in \cref{eqAuxiliary}. Let the sum of absolute values of the coefficients of the decomposed $C$ be $B = \sum_{\ell=-K}^{K} |c_\ell|$. We can find $\tilde \alpha: \{\tilde\alpha_m \in \mathbb C, \forall m\in [-T.. T]\}$ for estimator
$\tilde{x}(\alpha) = \sum_{m=-T}^T \alpha_{m}Q^m \bm b$ such that the following is satisfied
\begin{equation*}
\|C\tilde x (\tilde \alpha) - \bm b\| \le \min_{\alpha \in \mathbb C^{2t+1}} \|C\tilde x (\alpha) - \bm b\| + \varepsilon
\end{equation*}
using $\mathcal O (B^4(K+T)^2\|W\|\|W^{-1}\|^2(1+\|W^{-1}r\|^2/\varepsilon)$ measurements via the Hadamard test.
\end{lemma}

While \cref{algoAccess} can be dequantized with all quantum subroutines being discarded, we note that \cref{assumptionSample} is rather strong, and the preparation of $\mathcal{S}_{\bm b}$ may take up to $\mathcal O(N)$ time. Similarly, there is no guarantee that $\ket{b}$ can be effectively prepared on a quantum computer, hence the bottleneck of such algorithms lies in the preparation of the vector $\bm b$ for our algorithms, rather than the algorithms themselves. 

\bibliography{main}

\end{document}